\documentclass[a4paper,12pt]{amsart}
\usepackage[left=1.2in, right=1.2in, , bottom = 1.2in, top = 1.2in]{geometry}
\usepackage[utf8]{inputenc}
\usepackage[all]{xy}
\usepackage{amssymb}
\usepackage{amsmath}
\usepackage{csquotes}
\usepackage{enumerate}
\usepackage{graphicx}
\usepackage{enumerate}
\usepackage{amsaddr}
\usepackage{color}
\usepackage{mwe}
\usepackage{appendix}
\usepackage{arydshln}
\usepackage{bm}
\usepackage{mathtools}
\usepackage{subfigure}
\usepackage{subfig}
\usepackage{multirow}
\usepackage{rotating}
\usepackage[linesnumbered,ruled]{algorithm2e}
\usepackage{comment}

\usepackage{thmtools, thm-restate}

\usepackage{soul}
\usepackage{xcolor}
\newcommand\independent{\protect\mathpalette{\protect\independenT}{\perp}}
\def\independenT#1#2{\mathrel{\rlap{$#1#2$}\mkern2mu{#1#2}}}

\usepackage{tikz}
\usetikzlibrary{arrows.meta,shapes}
\usetikzlibrary{arrows,shapes.arrows,shapes.geometric,shapes.multipart,
decorations.pathmorphing,positioning,shapes.swigs,}
\usepackage{setspace}
\newtheorem{lemma}{Lemma}

\newtheorem{proposition}{Proposition}
\newtheorem{definition}{Definition}
\newtheorem{corollary}{Corollary}
\newtheorem{assumption}{Assumption}

\MakeOuterQuote{"}\EnableQuotes
\DeclareUnicodeCharacter{00A0}{ }

\title[]{Encoding and inference on separable effects for sustained treatments}

\author{Ignacio Gonz\'alez-P\'erez$^{1}$, Kerollos Nashat Wanis$^{2}$, Aaron Leor Sarvet$^{3}$, Mats Julius Stensrud$^{1}$} \address{ \small $^1$ Institute of Mathematics, \'Ecole Polytechnique Fédérale de Lausanne, Switzerland\\ $^{2}$ Departments of Breast Surgical Oncology and Health Services Research, and the Institute for Data Science in Oncology, The University of Texas MD Anderson Cancer Center, USA \\
$^{3}$ Department of Biostatistics \& Epidemiology, University of Massachusetts, Amherst, USA
}
\allowdisplaybreaks
\begin{document}
\newcommand{\cor}{\overline{c}=0,\overline{r}=1}

\begin{abstract}
Sustained treatment strategies are common in many domains, particularly in medicine, where many treatment are delivered repeatedly over time. The effects of adherence to a treatment strategy throughout follow-up are often more relevant to decision-makers than effects of treatment assignment or initiation. Here we consider the \emph{separable effect} of sustained use of a time-varying treatment. Despite the potential usefulness of this estimand, the theory of separable effects has yet to be extended to settings with sustained treatment strategies. To derive our results, we use an unconventional encoding of time-varying treatment strategies. This allows us to obtain concise formulations of identifying assumptions with better practical properties; for example, they admit frugal graphical representations and formulations of identifying functionals. These functionals are used to motivate doubly robust semiparametrically efficient estimators. The results are applied to the Systolic Blood Pressure Intervention Trial (SPRINT), where we estimate a separable effect of modified blood pressure treatments on the risk of acute kidney injury.
\end{abstract}

\maketitle

\textit{Keywords:} separable effects, sustained effects, sustained treatment encoding, adherence.

\section{Introduction}\label{sec:Intro} Treatments are often administered over time. Even in conventional randomized trials, where treatment strategies are assigned at baseline, the actual administration of treatment typically occurs sequentially. For example, the Systolic Blood Pressure Intervention Trial (SPRINT) \cite{sprint2015randomized} randomly assigned individuals to receive intensive or standard blood pressure therapy. While the assignment occurred at baseline, the therapy was meant to be taken daily in each arm. In trials like the SPRINT, where individuals' long-term behaviours may deviate from their assigned protocols, it is debated whether the effect of assignment, usually termed the ``intention-to-treat'' effect, is the most relevant quantity for decision-making \cite{hernan2017per,murray2021causal,hernan2012beyond,murray2016adherence}. An alternative is a particular sustained effect: an effect defined by interventions where individuals continuously take a single treatment throughout follow-up. 

We will consider effects of sustained treatment use in settings with competing events. To be concrete, in the SPRINT example, consider the objective of estimating the effect of intense blood pressure therapy on the risk of acute kidney injury (AKI), one of the trial’s secondary outcomes. Because AKI only occurs in individuals that are alive, death acts as a competing event. In this trial, the investigators found that the cumulative incidence of AKI is higher in the intensive blood pressure therapy arm compared to the control arm. Interpreting this finding as an ``effect'' of intensive blood pressure therapy on AKI etiology requires reasoning about causal mechanisms; the problem is that the drug can effect AKI in ways that may or may not be of interest. On the one hand, intensive blood pressure therapy may prevent AKI by binding directly to receptors in the kidneys. On the other hand, intensive blood pressure therapy may cause AKI by preventing  fatal events; when an individual survives such events, they might go on to experience an AKI that would otherwise be ``prevented'' by death. The confluence of these mechanisms complicates the interpretation of an overall effect of blood pressure therapy on AKI. 

Investigators have confronted these challenges by considering different estimands. These estimands, however, have their own interpretive complications: for example,  the controlled direct effect considers unrealistic interventions on the the competing event, here death, \cite{young2020causal} or the survivor average causal effect \cite{robins1986new} condition on an unmeasurable sub-population defined by cross-world counterfactuals. Here we will consider separable effects \cite{young2020causal, stensrud2021generalized, stensrud2022separable, stensrud2023conditional}. The motivation for these effects is to formalize notions of causal mechanisms that are practically relevant. For example,  investigators might be interested in the notion of ``effects not through the competing event''. Separable effects formalize this notion by defining modified versions of the original treatment. These modified treatments might be assumed to affect outcomes through specific pathways, which we will represent in causal graphs; in this case, they represent hypothetical refined medications that aim to isolate certain beneficial effects of the treatment of interest. By requiring an explicit, possibly unrealized modification of the treatment, separable effects help clarify the causal question under investigation and may also stimulate new ideas on how treatments could be improved in practice \cite{stensrud2021generalized, stensrud2022separable, stensrud2023conditional,robins2010alternative, robins2020interventionist}. Separable effects could, at least in principle, be directly computed in future clinical trials. Moreover, in observational data they usually rest on weaker assumptions for identification than previously proposed estimands in these settings \cite{young2020causal,robins1986new}.

So far separable effects have been considered in various settings with a single treatment assignment and failure time outcomes \cite{young2020causal, stensrud2021generalized, stensrud2022separable, stensrud2023conditional,robins2010alternative, robins2020interventionist,wanis2024separable, shpitser2017modeling, didelez2019defining, di2024longitudinal, maltzahn2024separable}. 
Alternatively, a sustained separable effect would enforce adherence to modified treatments \textit{continuously} over time. As suggested in \cite[Sec. 3.5(a)]{robins2020interventionist}, extending existing results on separable effects requires a causal model that incorporates treatment decomposition at each time point. The complexity of directly applying this causal model has challenged theoretical development for these important estimands \cite{robins2020interventionist}. 

We overcome these challenges by reformulating sustained separable effects via an alternative encoding, which is similar to an encoding used in the literature on the sustained effects of dynamic strategies \cite{robins1993analytic,hernan2000marginal,hernan2006comparison,hernan2018estimate}. Therein, the actual treatments an individual takes are entirely summarized by a variable denoting adherence to a protocol of interest over time. This encoding provides a parsimonious representation of assumptions and identifying functionals, which is particularly useful in the separable effects setting. The resulting expressions closely resemble those from the simpler point-treatment setting \cite{stensrud2021generalized,stensrud2022separable}, which reduce to special cases of our formalism. The encodings are discussed in more detail in Section \ref{sec:AdherenceEncodings}. Then, in Section \ref{sec:ObservedDataStructure} we explicitly introduce the observed data structure for a longitudinal setting with competing events. In Section \ref{sec:TreatmentDecomposition} we define the separable effects on the event of interest as contrasts of counterfactual hazards, and discuss their causal interpretation. Sufficient identification conditions are studied in Section \ref{sec:Identification}. This is followed by different estimation techniques in Section \ref{sec:Estimation}. The properties of these estimators are analysed in Section \ref{sec:SimulationStudy}, where we  perform a simulation study to assess coverage of bootstrap-based confidence intervals and robustness properties. In Section \ref{sec:SprintExample}, we apply our methods to the SPRINT, providing for the first time a sustained separable effects analysis of the results of this RCT.

\section{A comparison of different treatment encodings}\label{sec:AdherenceEncodings}

Sustained treatment strategies are routinely encoded by random variables that specify the treatment taken by an individual at each point in time. Foundational work on time-varying treatments \cite{robins1986new} introduced such encodings, which are now standard in causal inference textbooks \cite{hernan2018causal,VanDerLaan2011targeted}, and widely used in clinical trials \cite{jindani2023four,shehadeh2023dapagliflozin}, mediation analyses \cite{vanderweele2017mediation,zheng2017longitudinal,diaz2023efficient,lin2017mediation}, observational studies \cite{robins1993analytic,hernan2006comparison}, and the dynamic treatment regime literature \cite{tsiatis2019dynamic,murphy2003optimal}.

In this work, we highlight an alternative encoding of  treatment strategies that simplifies the formulation of estimands and identification results.
To be explicit, let $Z$ denote the treatment initiated at baseline. For example, in the SPRINT, $Z = 1$ denotes that an intensive blood pressure treatment strategy was initiated at the beginning of the study and $Z = 0$ denotes that standard  strategy was initiated. We define two  encodings as follows:

\begin{itemize}
    \item \textbf{Treatment-centered encoding:} Let $A_k \in \text{supp}(Z)$ denote the treatment actually taken by the individual at time $k$. This is a natural generalization of standard encodings used in static  treatment strategies. For example, an individual with $A_k = 1$ took intensive blood pressure control at time $k$. 
    \item \textbf{Strategy-centered encoding:} Let $R_k \in \{0, 1\}$ indicate whether an individual took the same treatment at time $k$ as the one given at baseline (that is, whether $A_k = Z$). For example, an individual with $(Z=1, R_k = 1)$ initiated intensive blood pressure therapy and subsequently  continued this treatment at time $k$.
\end{itemize}

Whereas the treatment-centered encoding is most commonly used in methodological work, the relevance of the strategy-centered encoding has been recognized in settings with complex dynamic strategies \cite{wanis2020adjusting,cain2010start} and applied to the  analysis of clinical trials \cite{murray2016adherence, lodi2021per}. The strategy-centered encoding is related to the ``clone-censor-weight'' estimation strategy popularized in translational causal inference literature for the  analysis of randomized trials \cite{gaber2024demystifying} and their ``emulations'' using observational data \cite{hernan2006comparison,cain2010start}. Therein, the strategy-centered encoding is derived and used to simplify the procedures applied to data \cite[App. 1]{cain2010start}. \cite{wanis2024separable} leveraged this encoding to study the effects of initiating a modified treatment. Here, we extend this to a sustained separable effect.

To make the relation between these encodings precise, consider the random vector $V \equiv (A_Y, A_D, D, Y)$, where $A_Y \equiv (A_{Y,1},\ldots,A_{Y,K+1})$, $A_D \equiv (A_{D,1},\ldots,A_{D,K+1})$, $D \equiv (D_1,\ldots,D_{K+1})$, and $Y \equiv (Y_1,\ldots,Y_{K+1})$. In our setting, $k \in \{1, \ldots, K+1\}$ indexes equally spaced discrete time intervals. The indicator $D_k$ denotes a post-treatment competing event occurring in interval $k$, and $Y_k$ is the indicator of the primary event of interest.  In subsequent sections, $A_{Y}$ and $A_{D}$ will represent components of the original treatment variable $A$, which are deterministically related in the observed data. To simplify exposition in this section, we take them to be two separate treatment vectors.

We will use the treatment-centered encoding to represent variables within a non-parametric structural equation model (NPSEM) \cite{robins2010alternative, richardson2013single, pearl}, e.g., a Finest Fully Randomized Causally Interpretable Structured Tree Graphs (FFRCISTG) model \cite{robins1986new,richardson2013single}. Based on this encoding, we define an alternative random vector $V^* \equiv (Z_Y, Z_D, R, D, Y)$. Here, $Z_Y \equiv A_{Y,1}$, $Z_D \equiv A_{D,1}$ represent initiated treatment and $R_k \equiv I(A_{Y,k} = Z_Y, A_{D,k} = Z_D)$ represents adherence to the initiated treatment at time $k$, $R\equiv (R_1,\ldots,R_{K+1})$.  Then, $V^*$ is a \textit{strategy-centered encoding} that allows us to re-express adherence in terms of agreement with the initiated treatment. Specifically, perfect adherence to the initiated treatment is represented as $R_k=1$ for all $k$. 

We use directed acyclic graphs (DAG) to represent a statistical model induced by an FFRCISTG \cite{robins1986new,richardson2013single}. An example DAG for $V$, which we call a \textit{treatment-centered} DAG, is given in Figure \ref{fig:ExampleTwoEncogingDAGs} (a). Due to the deterministic relations between $V$ and $V^*$, a statistical model for $V$ will imply a statistical model for $V^*$ that can itself be represented by a DAG. This \textit{strategy-centered} DAG is obtained by a simple graphical algorithm: for each node $X$ that is a child of $A_{Y,k}$ in the treatment-centered DAG $\mathcal{G}$, add edges $Z_Y\to X$ and $R_k\to X$ in a strategy-centered DAG $\mathcal{G}^*$. For each node $W$ that is a parent of $A_{Y,k}$ in $\mathcal{G}$, add the edge $W\to R_k$ in $\mathcal{G}^*$.  The same procedure follows for $A_{D,k}$, $Z_D$ and $R_k$. Lastly, all other edges in $\mathcal{G}$ among the variables $V \setminus \{A_Y,A_D\}$ are retained in $\mathcal{G}^*$. The strategy-centered DAG in Figure \ref{fig:ExampleTwoEncogingDAGs} (b) is the  counterpart of the treatment-centered DAG in 
Figure \ref{fig:ExampleTwoEncogingDAGs} (a). The general algorithm is detailed in Appendix \ref{sec:AppAlgoDrawDags}. As we elaborate in Section \ref{sec:Identification} and Appendix \ref{sec:AppEquivalenceIdentificationConditions}, strategy-centered DAGs permit more-parsimonious depictions of the assumptions used to identify separable effects.

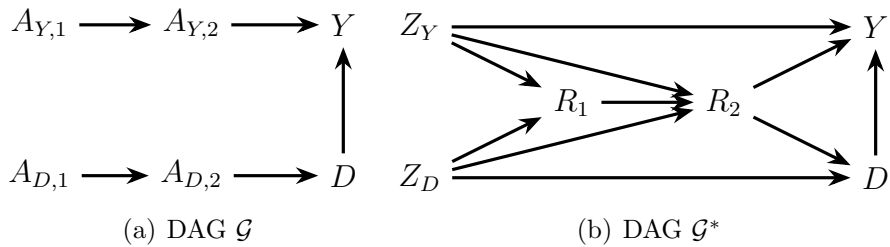
\begin{figure}
    \centering
    \subfigure[DAG $\mathcal{G}$]{{\begin{tikzpicture}
\begin{scope}[every node/.style={thick,draw=none}]
    \node (AY1) at (-1,1) {$A_{Y,1}$};
    \node (AD1) at (-1,-1) {$A_{D,1}$};
    \node (AY2) at (1,1) {$A_{Y,2}$};
    \node (AD2) at (1,-1) {$A_{D,2}$};
    \node (Y) at (3,1) {$Y$};
    \node (D) at (3,-1) {$D$};
\end{scope}

\begin{scope}[>={Stealth[black]},
              every node/.style={fill=white,circle},
              every edge/.style={draw=black,very thick}]
	\path [->] (AY1) edge (AY2);
    \path [->] (AD1) edge (AD2);
    \path [->] (AY2) edge (Y);
    \path [->] (AD2) edge (D);
    \path [->] (D) edge (Y);

\end{scope}
\end{tikzpicture}
}}    
\subfigure[DAG $\mathcal{G}^*$]{{\begin{tikzpicture}
\begin{scope}[every node/.style={thick,draw=none}]
    \node (ZY) at (-1,1) {$Z_Y$};
    \node (ZD) at (-1,-1) {$Z_D$};
    \node (R1) at (1,0) {$R_1$};
    \node (R2) at (3,0) {$R_2$};
    \node (Y) at (5,1) {$Y$};
    \node (D) at (5,-1) {$D$};
\end{scope}

\begin{scope}[>={Stealth[black]},
              every node/.style={fill=white,circle},
              every edge/.style={draw=black,very thick}]
    \path [->] (ZY) edge (R1);
    \path [->] (ZY) edge (R2);
    \path [->] (ZY) edge (Y);
    \path [->] (ZD) edge (R1);
    \path [->] (ZD) edge (R2);
    \path [->] (R1) edge (R2);
    \path [->] (ZD) edge (D);
    \path [->] (R2) edge (Y);
    \path [->] (R2) edge (D);
    \path [->] (D) edge (Y);

\end{scope}
\end{tikzpicture}
}}    
    
\caption{\small Example of a DAG $\mathcal{G}$ under the primitive (treatment-centered) encoding (a), and its associated DAG $\mathcal{G}^*$ under the strategy-centered encoding (b), following Algorithm \ref{alg:DagDrawing} (see Appendix \ref{sec:AppAlgoDrawDags}). }
\label{fig:ExampleTwoEncogingDAGs}
\end{figure}

\section{Observed data structure and notation}
\label{sec:ObservedDataStructure}
Consider a study with $n$ i.i.d.\ individuals who are taking a binary treatment sequentially at times $k\in\{1,\ldots, K+1\}$. Let the treatment taken at the first time point be $Z\in \{0,1\}$. 
For each individual, a vector $L_{0}$  of  pre-randomization (baseline) covariates is measured.  Let $D_k,Y_k$ be defined as in Section \ref{sec:AdherenceEncodings}, and denote by $R_k$ the strategy-centered indicator of adherence to the initiated treatment during the time interval $k$.  Let $C_k$ denote the indicator of loss of follow-up at time $k$. Furthermore, let $L_k$ denote a vector of post-randomization covariates at time point $k$. We will use overlines to denote the history of a random variable through $k$, e.g., $\overline{D}_{k} \equiv (D_0,\dots,D_{k} $), and underlines to denote its future relative to $k+1$, e.g., $\underline{D}_{k+1} \equiv (D_{k+1},\dots,D_{K+1}$). 

In Sections \ref{sec:Identification} to \ref{sec:Estimation} we will assume that all the supports of the time-varying covariates $L_k$ are countable (possibly infinite) sets, but extensions to the absolutely continuous case are straightforward. We assume that at each time $D_k$ is recorded before $Y_k$. This is motivated by our running example on the SPRINT, where $Y_k$ represents AKI, and the competing event $D_k$ indicates death from any cause. If the individual is event-free at time $k$, meaning that $Y_k=D_k=0$, their post-treatment covariates $L_k$ are recorded. In our running example of AKI, these would include blood-pressure measurements at every time-point.  If an individual experiences a competing event at time $k$ without a history of the event of interest, then this event can no longer occur: $D_k=1,\ Y_{k-1}=0$ implies $Y_j=0$ for all $j\geq k$. In addition, we will consider $Y, D,C$ to be monotone increasing, in the sense that $Y_k=1$ implies $Y_j=1$ for all $j>k$, and similarly for $D,C$. We let $Y_{K+1}$ be the last recorded outcome, and we suppose that all individuals are event-free and uncensored  prior to randomization (which again is reasonable in the SPRINT example): $D_0\equiv Y_0\equiv C_0\equiv 0$. Therefore, our temporal convention is $(L_{0},Z, \dots,C_k,R_k,D_{k},Y_k,L_{k},\dots,D_{K+1},Y_{K+1})$.

\section{Treatment decomposition and separable effects}\label{sec:TreatmentDecomposition}
Consider a four-arm trial where each arm initiates patients to a possible combination of the binary treatments  $A_{Y,1},A_{D,1}\in\{0,1\}$. Following \cite{robins2010alternative} in a point treatment setting, suppose that $A_{Y,k},A_{D,k}$ constitute a decomposition of a binary treatment $A_k$, such that taking $A_k=z$ in the two-arm trial is equivalent to taking $A_{Y,k}=A_{D,k}=z$ in the four-arm trial. In \cite{stensrud2021generalized}, these assumptions are referred to as a \textit{generalized decomposition assumption}, here extended to a longitudinal setting. 
This equivalence also follows from the more general \textit{modified treatment assumption} in \cite{stensrud2023conditional}. Take our running example of the SPRINT. Its intensive treatment arm targets a systolic blood pressure of 120 mm Hg, and 140 mm Hg for the standard therapy. As explained in  \cite{sprint2015randomized} ``Medications for participants in the intensive-treatment group were adjusted on a monthly basis to target a systolic blood pressure of less than 120 mm Hg. For participants in the standard treatment group, [...] the dose was reduced if systolic blood pressure was less than 130 mm Hg on a single visit or less than 135 mm Hg on two consecutive visits''. While the medication an individual is told to take might vary over time, the treatment protocol they are assigned at baseline remains fixed throughout the trial. This protocol prescribes a treatment at each time point. This is an example of what \cite{young2019inverse} name deterministic time-varying strategies, which will be the focus of our study.

In the four-arm trial, we will hereby use the strategy-centered encoding where $(Z_Y,Z_D)$ represent the  initiated treatment components at baseline and  $R_k$ indicates whether an individual adhered to the initiated treatments at time $k$, that is, whether $(A_{Y,k},A_{D,k}) = (Z_Y,Z_D)$, as introduced in Section \ref{sec:AdherenceEncodings}. Thus, instead of explicitly writing  decomposed treatment strategies over time, the time-varying indicator $R_k$ is a simple binary variable. Had we chosen a treatment-centered encoding, we would have also decomposed the $A_k$ at every time $k$ (as per Section \ref{sec:AdherenceEncodings} and Appendix \ref{sec:EquivalenceAdherenceEncodings}). In the SPRINT example, the blood pressure treatment $Z$ could be decomposed into two components: $Z_Y$, which is supposed to have an effect on AKI by binding directly to receptors in the kidneys and relaxing the efferent arterioles, a known effect of blood pressure drugs such as angiotesin II receptor blockers; and $Z_D$, including all remaining components of the treatment that, e.g., affects systemic blood pressure. Inference on the effects of such a modified treatment could support the development of an improved drug that avoids direct effects on the kidneys ($Z_D = 1, Z_Y = 0$).

Our main goal is to study the effect of a sustained treatment strategy on the event of interest. Let $Y_{k+1}^{z_Y,z_D,\cor}$  an individual's indicator of the event of interest at time $k+1$ when, possibly contrary to fact, $Z_Y$ is set to $z_Y$, $Z_D$ to $z_D$, the individual is always followed-up, and perfectly adheres to the (possibly modified) treatment. By definition, $Y_{k+1}^{z_Y,z_D,\cor}$ is equivalent to $Y_{k+1}^{\overline{A}_{Y,k+1}= z_Y,\overline{A}_{D,k+1}=z_D,\overline{c}=0}$ under the treatment-centered encoding.  
We consider the following separable effects: 
\begin{definition}[Sustained $Z_Y$ separable effect]\label{def:AYSeparableEffect}
    The\textit{ sustained $Z_Y$ separable effect} evaluated at $z_D\in\{0,1\}$ is the contrast
    \begin{equation}\label{eq:AYSeparableEffect}
        \mathbb{P}(Y_{K+1}^{z_Y=1,z_D,\cor}=1) \text{ versus } \mathbb{P}(Y_{K+1}^{z_Y=0,z_D,\cor}=1).
    \end{equation}
\end{definition}

\begin{definition}[Sustained $Z_D$ separable effect ]\label{def:ADSeparableEffect}
    The \textit{ sustained $Z_D$ separable effect}  evaluated at $z_Y\in\{0,1\}$ is the contrast
    \begin{equation}\label{eq:ADSeparableEffect}
        \mathbb{P}(Y_{K+1}^{z_Y,z_D=1,\cor}=1) \text{ versus } \mathbb{P}(Y_{K+1}^{z_Y,z_D=0,\cor}=1).
    \end{equation}
\end{definition}
The contrast in Equation \eqref{eq:AYSeparableEffect} quantifies the causal effect of the $Z_Y$ component on the probability of observing the event of interest under an intervention that fixes $Z_Y$ to $z_Y$, and imposes no loss of follow-up and perfect adherence to the initiated treatment. The contrast in Equation \eqref{eq:ADSeparableEffect} can be interpreted similarly. 
 
So far we have made no assumptions about how $Z_Y$ and $Z_D$ affect the competing event and the event of interest. Under certain conditional independence assumptions, however, we can give specific interpretations to the separable effects. We will use causal DAGs to describe such assumptions, representing an underlying FFRCISTG \cite{robins1986new,richardson2013single}. This  counterfactual causal model makes strictly fewer assumptions than the  non-parametric structural equation model with independent errors (NPSEM-IE) \cite{ robins2010alternative, richardson2013single, pearl}. The absence of an arrow in a causal DAG representing a FFRCISTG model can either encode (i) the assumption that an individual level causal effect is absent for every individual in the study population or (ii) the weaker assumption that a population level causal effect is absent \cite{robins2020interventionist, richardson2013single, dawid2015statistical, dawid2000causal}.  
To represent the causal relations and assumptions of the treatment decomposition in the observed data, we will use an \textit{extended} causal DAG \cite{robins2010alternative,robins2020interventionist}. This graph includes both the modified treatment nodes, $Z_Y$ and $Z_D$,  the original treatment node $Z$, and bold edges  $Z \boldsymbol{\rightarrow} Z_Y$  and $Z \boldsymbol{\rightarrow} Z_D$ representing the deterministic relations between the full treatment $Z$ and these two components in the observed data. 

Earlier in this section we motivated $Z_Y, Z_D$ as two treatment components that exert effects through different causal paths.\footnotemark[1] This can be formalized by extending the notion of partial isolation \cite{stensrud2021generalized} to our context: 
\footnotetext[1]{In a DAG over nodes $V$  we say that a path from $X\subset V$  to  $B\in V$  is causal if it is a directed path which does not intersect $X$ after the first node. This notion also applies to the extended DAG, where we discard the deterministic treatment edges. We say that a path intersects  $U\subset V$ when it contains one of the nodes in $U$, not being the terminal nodes of the path.}
\begin{definition}[$Z_Y$ partial isolation]\label{def:AYPI}
    All causal paths from $Z_Y$ to  any node in $\overline{D}_{K+1}$ in the extended DAG intersect $\overline{Y}_{K}$, $\overline{R}_{K+1}$ or $\overline{C}_{K+1}$.
\end{definition}
\begin{definition}[$Z_D$ partial isolation]\label{def:ADPI}
    All causal paths from $Z_D$ to  any node in $\overline{Y}_{K+1}$ in the extended DAG intersect $\overline{D}_{K+1}$, $\overline{R}_{K+1}$, or $\overline{C}_{K+1}$.
\end{definition}
If both $Z_Y$ and $Z_D$ partial isolation hold, we say that \textit{full isolation} holds. Partial isolation conditions are not required for identification of separable effects, see Section \ref{sec:Identification}. These definitions differ from those in \cite{stensrud2021generalized}. For example, the conditions in \cite{stensrud2021generalized} are violated by a path from $Z_D$ to $Y_{k}$ not intersected by $\overline{D}_{k}$; such a path does not violate the condition of Definition \ref{def:ADPI}, so long as it intersected $\overline{R}_{k}$ or $\overline{C}_{k}$.  The partial isolation conditions presented here are different from those in \cite{stensrud2021generalized} because we consider an intervention that additionally fixes censoring and adherence, thereby interrupting causal paths involving those variables.\footnotemark[2] Nevertheless, the conditions here and those in \cite{stensrud2021generalized} give interpretations to separable effects as direct and indirect effects. Under $Z_Y$ partial isolation, the $Z_Y$ separable effect is interpretable as a direct effect, because it would only involve pathways outside of the competing event. Likewise, under $Z_D$ partial isolation the $Z_D$ separable effect is interpretable as an indirect effect,  because it would only involve pathways mediated by the competing event.

\footnotetext[2]{Any directed path in the extended causal DAG intersecting $\overline{R}_{K+1}$ or $\overline{C}_{K+1}$ would not exist in the SWIG associated with the intervention $\{\overline{R}_{K+1}=1,\overline{C}_{K+1}=0\}$ due to the node splitting.}

In the SPRINT, if we accept that  $Z_Y$  has no effect on blood pressure except its possible effect on AKI, such that the only direct effect on blood pressure and death from any cause comes from  $Z_D$, then $Z_Y$ partial isolation would hold. However, it is unlikely that $Z_D$ partial isolation would also hold in this example, as a reduction in blood pressure can increase the AKI risk. This example also illustrates how the partial isolation conditions (Definitions \ref{def:AYPI}-\ref{def:ADPI}) are usually accepted or rejected based on subject matter knowledge.

\section{Identification of separable effects}
\label{sec:Identification}
To study the sustained separable effects introduced in Section \ref{sec:TreatmentDecomposition}, we need to identify $\mathbb{P}(Y_{K+1}^{z_Y,z_D,\cor}=1)$. Regardless of any isolation conditions, if we had access to a four-arm trial in which patients randomly initiated some level of $Z_Y$ and $Z_D$, there was no loss of follow-up and each individual adhered perfectly to their initial treatment, we would be able to identify the target quantity by $\sum_{k=0}^K\mathbb{P}(Y_{k+1}=1\mid Z_Y=z_Y, Z_D=z_D,D_{k+1}=Y_k=0)$ \cite{hernan2018causal}: the required conditions (conditional exchangeability, positivity and consistency) would hold by design. Thus, this is a single-world estimand that, in principle, can be identified from a perfectly designed randomized experiment. When only data from the two-arm trial are available (without arms where $Z_Y\neq Z_D$), identification of $\mathbb{P}(Y_{K+1}^{z_Y,z_D,\cor}=1)$ is not guaranteed when $z_Y\neq z_D$, even under no loss of follow-up and perfect adherence. Here we will give further assumptions that are sufficient for identification in this setting, which resembles the observed data. First we give conventional assumptions of conditional exchangeability, consistency and positivity. 

\begin{assumption}[Conditional exchangeability]\label{ass:Exchangeability}
For all $z$ in $\{0,1\}$,
\begin{subequations}\label{eq:Exchangeability}
    \begin{equation}\label{eq:ExchangeabilityPast}
        \overline{Y}_{K+1}^{z,\cor},\ \overline{D}_{K+1}^{z,\cor},\ \overline{L}_{K}^{z,\cor}\independent Z\ |\ L_0,
    \end{equation}
    and for each $k\in\{0,\ldots,K\}$
    \begin{equation}\label{eq:ExchangeabilityFuture}
        \begin{split}
            \underline{Y}_{k+1}^{z,\cor},\ \underline{D}_{k+1}^{z,\cor},\ \underline{L}_{k+1}^{z,\cor}\independent C_{k+1},\ R_{k+1}\ |&\{ Y_k=D_k=0,\\& \overline{L}_k,\ Z=z,\ \overline{C}_k=0,\ \overline{R}_{k}=1.\}
        \end{split}
    \end{equation}
\end{subequations}
\end{assumption}

\begin{assumption}[Consistency]\label{ass:Consistency}
    For each $k$, if $Z=z$, ${C}_k=0,$ and $ \overline{R}_{k}=1,$ then
    \begin{equation}\label{eq:Consistency}
        \overline{Y}_{k}^{z,\cor}=\overline{Y}_{k},\ \overline{D}_{k}^{z,\cor}=\overline{D}_{k} \text{ and } \overline{L}_{k}^{z,\cor}=\overline{L}_{k}.
    \end{equation}
\end{assumption}
This consistency condition is conventional in the causal inference literature. The conditional exchangeability conditions \eqref{eq:Exchangeability} could be expanded into $6(K+1)$ sequential conditional independencies which are slightly weaker than the ones we present here.
\begin{assumption}[Positivity]\label{ass:Positivity}
For $z$ in $\{0,1\}$ and $k\in\{0,\ldots,K\}$

\begin{subequations}\label{eq:Positivity}
    \begin{align}
    &\mathbb{P}(L_0=l_0)>0\Rightarrow\mathbb{P}(Z=z|L_0=l_0)>0,\label{eq:PositivityL0}\\
    &\mathbb{P}(\overline{L}_{k}=\overline{l}_{k}, \overline{R}_{k+1}=1,Y_k=D_{k+1}=C_{k+1}=0)>0\Rightarrow\label{eq:PositivityA}\\\nonumber&\qquad \mathbb{P}(Z=z|\overline{L}_{k}=\overline{l}_{k}, \overline{R}_{k+1}=1,Y_k=D_{k+1}=C_{k+1}=0)>0,\\
    &\mathbb{P}(\overline{L}_{k}=\overline{l}_{k}, \overline{R}_{k}=1, Z=z,Y_k=D_{k}=C_{k}=0)>0\Rightarrow\label{eq:PositivityCR}\\\nonumber&\qquad \mathbb{P}(R_{k+1}=1,C_{k+1}=0|\overline{L}_{k}=\overline{l}_{k}, \overline{R}_{k}=1, Z=z,Y_k=D_{k}=C_{k}=0)>0.
\end{align}
\end{subequations}

\end{assumption}
Assumption \eqref{eq:PositivityL0} is the standard positivity conditions under interventions on $Z$. Assumption \eqref{eq:PositivityA} requires that, at any time point, for any possible history of the time varying covariates under survival, perfect adherence and no loss of follow-up, individuals which initiated both treatments  can be found. Condition \eqref{eq:PositivityCR} states that for any treatment and covariate history, under survival, no loss of follow-up and adherence up to a certain time point, individuals can be found which adhere to treatment and are followed-up to the next time.

\begin{assumption}[Dismissible component conditions]\label{ass:DCC}
Let the time varying covariates be expressed as two components: $L_k=(L_{D,k},L_{Y,k})$. Furthermore, let $G$ refer to  the four-arm trial where $Z_Y$ and $Z_D$ are randomly assigned, but the causal structure between variables is otherwise identical to the observed data. We use the notation $X(G)$ to indicate  a variable $X$ in this trial. Then, for all $k \in \{0,\dots, K\}$
\begin{align}\label{eq:DCC}
     Y_{k+1}^{\cor}(G) \independent Z_D(G) \mid& Z_Y(G), D_{k+1}^{\cor} (G)=Y_k^{\cor}(G)=0, \overline{L}_k^{\cor}(G), \\
      D_{k+1}^{\cor}(G) \independent Z_Y(G) \mid& Z_D(G), D_{k}^{\cor} (G)=Y_k^{\cor}(G)=0, \overline{L}_k^{\cor}(G),\nonumber\\
     L_{Y,k}^{\cor}(G) \independent Z_D(G) \mid& \{Z_Y(G), D_{k}^{\cor} (G)=Y_k^{\cor}(G)=0, L_{D,k}^{\cor}(G),\nonumber\\ & \overline{L}_{k-1}^{\cor}(G)\}, \nonumber\\
     L_{D,k}^{\cor}(G) \independent Z_Y(G) \mid& Z_D(G), D_{k}^{\cor} (G)=Y_k^{\cor}(G)=0,  \overline{L}_{k-1}^{\cor}(G),  \nonumber
\end{align} 
 \end{assumption}

The dismissible component conditions are closely related to the notions of Partial isolation (Definitions \ref{def:AYPI}-\ref{def:ADPI}):
\begin{restatable}{proposition}{propDccImpliesAYPI}\label{prop:DccImpliesAYPI}
    If the dismissible component conditions (Assumption \ref{ass:DCC}) hold under a partition where $\overline{L}_{Y,K}=\emptyset$, then $Z_Y$ partial isolation holds.
\end{restatable}
We provide the proof of this proposition, together with other results   related to $Z_D$ and full isolation (Proposition \ref{prop:DccImpliesADPI} and Corollary \ref{coro:DccImpliesFullPI}), in Appendix \ref{sec:AppProofDCCimpliesPI}.

Conditions \eqref{eq:ExchangeabilityPast}, \eqref{eq:Consistency}, \eqref{eq:PositivityL0}   are expected to hold by design when $Z$ is randomly assigned. Conditions \eqref{eq:PositivityCR}  and  \eqref{eq:ExchangeabilityFuture} does not hold by design as censoring and adherence are not assigned at random. Yet, the former could be tested in the observed data. The dismissible component conditions \eqref{eq:DCC} and the positivity condition \eqref{eq:PositivityA} do not hold by design. However,   Assumption \ref{ass:DCC}  is single world and could be read off the SWIG associated to the extended causal DAG of the process under interventions on $Z_Y, Z_D$, perfect adherence and no loss of follow-up \cite{richardson2013single}.

Having stated these four sets of conditions, we can introduce the main identifiability result:
\begin{restatable}{theorem}{thmIdentifcationFormula}\label{thm:IdentifcationFormula}
Suppose that Assumptions \ref{ass:Exchangeability}-\ref{ass:DCC} hold under a FFRCISTG model. Then, the counterfactual probability of observing the event of interest is identified by
\begin{align}\label{eq:IdentificationFormula}\allowdisplaybreaks
        &\mathbb{P}(Y_{K+1}^{z_Y,z_D,\cor}=1)\\
        \nonumber=&\sum_{j=0}^{K}\sum_{\overline{l}_K}  \mathbb{P}  ( Y_{j+1}=1 \mid D_{j+1}=Y_j=C_{j+1}=0,  \overline{L}_{j}=\overline{l}_{j}, Z=z_Y, \overline{R}_{j+1}=1)  \\
        &  \prod_{s=0}^{j} \bigg[ \mathbb{P} ( D_{s+1}=0 \mid Y_s=D_{s}=C_{s+1}=0, \overline{L}_{s} = \overline{l}_{s}, Z=z_D, \overline{R}_{s+1}=1) \nonumber \\
        & \times \mathbb{P} ( L_{Y,s}=l_{Y,s} \mid Y_s=D_{s}=C_{s}=0, L_{D,s} =  l_{D,s}, \overline{L}_{s-1} = \overline{l}_{s-1}, Z =z_Y, \overline{R}_{s}=1)  \nonumber\\
        & \times \mathbb{P} ( L_{D,s}=l_{D,s} \mid Y_s=D_{s}=C_{s}=0, \overline{L}_{s-1} = \overline{l}_{s-1}, Z=z_D, \overline{R}_{s}=1)\nonumber\\
        & \times \mathbb{P}  ( Y_{s}=0 \mid D_{s}=Y_{s-1}=C_{s}=0,  \overline{L}_{s-1}=\overline{l}_{s-1}, Z=z_Y, \overline{R}_{s}=1) \bigg].\nonumber  
\end{align}
\end{restatable}
See Appendix \ref{sec:AppProofIdentification} for a proof. We refer to Equation \eqref{eq:IdentificationFormula} as the g-formula for $\mathbb{P}(Y_{K+1}^{z_Y,z_D,\cor}=1)$ \cite{robins1986new}, which is expressed in terms of factual quantities. We present equivalent conditions and results using the treatment-centered encoding in Appendix \ref{sec:EquivalenceAdherenceEncodings}.

\section{Estimation}\label{sec:Estimation}
Denote by $\nu(\mathbb{P})$ the right-hand-side of the identification formula \eqref{eq:IdentificationFormula} in Theorem \ref{thm:IdentifcationFormula}. We propose three  estimators of this quantity.
\subsection{Simple plug-in estimator}\label{sec:NaiveEstimator}
Let $\Tilde{\mathbb{P}}$ denote parametric models for the conditional distributions of $Y, D, L_D, L_Y$, as described in Theorem \ref{thm:IdentifcationFormula}. Let $\Tilde{\mathbb{P}}_n$ be the corresponding models fitted to the observed data. These models allow for the construction of a simple plug-in estimator, $\widehat{\nu}_{simple} = \nu(\Tilde{\mathbb{P}}_n)$. If the models $\Tilde{\mathbb{P}}$ are correctly specified and consistently estimated, $\widehat{\nu}_{simple}$ is a consistent estimator of $\mathbb{P}(Y_{K+1}^{z_Y,z_D,\cor}=1)$ under the identification conditions in Theorem \ref{thm:IdentifcationFormula}.
\subsection{Weighted estimator}\label{sec:WeightedEstimators}
To motivate weighted estimators, which are easier to fit and rely on less parametric assumptions than the one in Section \ref{sec:NaiveEstimator}, consider the following Theorem. 
\begin{restatable}{theorem}{thmWeigthedFormulas}\label{thm:WeigthedFormulas}
Under the conditions of Theorem \ref{thm:IdentifcationFormula} an  equivalent identification formula is
    \begin{align}\label{eq:WeigthedFormulaD}
        &\mathbb{P}(Y_{K+1}^{z_Y,z_D,\cor}=1)\\
         =& \sum_{s=0}^K\mathbb{E}  [ W_{(C,R),s}(z_D) W_{Y,s}  W_{L_{Y},s}  (1-Y_s)(1-D_{s+1}) Y_{s+1} \mid Z=z_D]\nonumber,
    \end{align}
where 
\begin{align*}
 W_{Y,s}&=\frac{  \mathbb{P}(Y_{s+1}=1 \mid C_{s+1}=D_{s+1}=Y_{s}= 0, \overline{L}_{s},  Z = z_Y,\overline{R}_{s}=1) }{   \mathbb{P}(Y_{s+1}=1 \mid C_{s+1}=D_{s+1}=Y_{s}= 0, \overline{L}_{s},  Z = z_D,\overline{R}_{s}=1) }\\
 &\times\prod_{j=0}^{s} \frac{  \mathbb{P}(Y_{j}=0 \mid C_{j}=D_{j}=Y_{j-1}= 0, \overline{L}_{j-1},  Z = z_Y,\overline{R}_{j-1}=1) }{   \mathbb{P}(Y_{j}=0 \mid C_{j}=D_{j}=Y_{j-1}= 0, \overline{L}_{j-1},  Z = z_D,\overline{R}_{j-1}=1) }, \\
  W_{L_{Y},s}  &=\prod_{j=0}^{s}  \frac{  \mathbb{P}(L_{Y,j} = l_{Y,j} \mid C_{j}=  D_{j} =Y_j= 0, L_{D,j}, \overline{L}_{j-1},  Z = z_Y, \overline{R}_{j-1}=1) }{  \mathbb{P}(L_{Y,j} = l_{Y,j} \mid C_{j}=  D_{j} =Y_j= 0, L_{D,j}, \overline{L}_{j-1},  Z = z_D, \overline{R}_{j-1}=1) }, \\
W_{(C,R),s} (z) &= \frac{I(C_{s+1} =0)I( \overline{R}_{s+1}=1) }{ \prod\limits_{j=0}^{s}  \mathbb{P}(C_{j+1}=0, R_{j+1}=1 \mid C_{j}=D_{j}=Y_j=0, \overline{L}_{j},  Z=z, \overline{R}_{j}=1) }.
\end{align*}
\end{restatable}

See Appendix \ref{sec:AppProofEstimation} for a proof. Equation \eqref{eq:WeigthedFormulaD} shows how sustained effect estimators implicitly perform artificial censoring \cite{robins1993analytic} of non-adherers, by the indicator functions in the numerator of $W_{(C,R),s}$. The function weights are non-zero so long as an individual took treatment ``consistent with having followed [the] regime'':  perfect adherence to the initiated treatment  \cite[App. 1]{cain2010start}.

Equation \eqref{eq:WeigthedFormulaD} motivates a weighted estimator based on parametric models $\Tilde{\mathbb{P}}$ for the quantities involved in the weights $ W_Y, W_{L_Y}, W_{(C,R)}$:
\begin{equation*}
    \widehat{\nu}_{weighted,Y}=\widehat{\mathbb{E}}_n\left[  \frac{I(Z=z_D)}{\widehat{\mathbb{P}}_n(Z=z_D)}\sum_{s=0}^K\Tilde{W}_{(C,R),s}(z_D) \Tilde{W}_{Y,s}  \Tilde{W}_{L_{Y},s}  (1-Y_s)(1-D_{s+1}) Y_{s+1}\right],
\end{equation*}
where $\widehat{\mathbb{P}}_n(Z=z_D)=\widehat{\mathbb{E}}_n[I(Z=z_D)]$ and $\widehat{\mathbb{E}}_n$ denotes the empirical mean. Under the conditions of Theorem \ref{thm:IdentifcationFormula},   $\widehat{\nu}_{weighted,Y}$ is a consistent estimator of the target quantity whenever and the parametric models for  $\{C_{j}=0, R_j=1\}$ as well as  $Y, L_Y$ are correctly specified and consistently estimated. 
Another weighted estimator which requires specification of the conditional distributions of $D,L_D$ instead of $Y, L_Y$ is presented in Appendix \ref{sec:AppEstimationWeighted}.

\subsection{Doubly robust estimator}\label{sec:DRMainBody}
A one-step estimator of $\nu(\mathbb{P})$ has  certain robustness guarantees, unlike the previous estimators. Specifically, define $\widehat{\nu}_{DR} =\widehat{\mathbb{E}}_n[\nu^1(\Tilde{\mathbb{P}}_n)]+\nu(\Tilde{\mathbb{P}}_n)$, where  $\Tilde{\mathbb{P}}_n$ is the law denoting the estimations of the  modelled quantities $\Tilde{\mathbb{P}}$, $\nu(\Tilde{\mathbb{P}}_n)$ is the simple plug-in estimator introduced in Section \ref{sec:NaiveEstimator} and $\nu^1$ is $\nu$'s influence function. The expression of this influence function can be seen in Theorem \ref{thm:InfluenceFucntion} in Appendix \ref{sec:AppEstimationDR}. We  show that this estimator is doubly robust (Theorem \ref{thm:DR} in Appendix \ref{sec:AppEstimationDR}), in the following sense: under the identification conditions of Theorem \ref{thm:IdentifcationFormula}, if the propensity score for $\{C,R\}$ is correctly specified, then the one-step estimator estimator will be consistent either if the models for the  conditional distributions of  $\overline{Y}_{K+1},\overline{L}_{Y,K}$ or for $\overline{D}_{K+1},\overline{L}_{D,K}$ are correctly specified, but not necessarily both. We also discuss how this one-step estimator can be adapted to construct doubly robust estimators of estimands presented in related work \cite{stensrud2021generalized,stensrud2022separable,wanis2024separable}. See Appendix \ref{sec:AppRelatedWorkDR} for details.

\section{Simulation study}\label{sec:SimulationStudy}
To verify nominal coverage of bootstrap-based confidence intervals, and our robustness statements for the one-step estimator (Theorem \ref{thm:DR} in Appendix \ref{sec:AppEstimationDR}), we conducted a simulation study. 

Consider the data structure introduced in Section \ref{sec:ObservedDataStructure} with $K=1$, that is, with two time points.  Details on the data generating mechanism are given in Appendix \ref{sec:AppSimus}. Under this data-generating process, the identifiability conditions required by Theorem \ref{thm:IdentifcationFormula} are satisfied (with $L_Y=\emptyset$), and thus the  probabilities $\mathbb{P}(Y_{2}^{z_Y,z_D,\cor}=1)$ are identified by the g-formula \eqref{eq:IdentificationFormula}. 
These true counterfactual probabilities are given in Table \ref{tab:TrueProbs}.

\begin{table}[ht]
    \centering
    \begin{tabular}{l|cccc}
        \hline
        $\bm{(z_Y, z_D)}$ & $\bm{(1,1)}$ & $\bm{(1,0)}$ & $\bm{(0,1)}$ & $\bm{(0,0)}$ \\ \hline
        $\mathbb{P}(Y_{2}^{z_Y,z_D,\cor}=1)$ \rule{0pt}{3ex} & 0.72 & $0.74$ & 0.62 & $0.66$ \\ 
    \end{tabular}
    \caption{True probabilities of observing the target event in the four-arm trial described in Section \ref{sec:SimulationStudy}, computed following \eqref{eq:IdentificationFormula}.}
    \label{tab:TrueProbs}
\end{table}

To assess coverage of the bootstrap confidence intervals, we generated, for sample sizes $n\in\{1000,5000,10000\}$ (in the two-arm trial), 500 trial realizations. Then, 95\%  confidence intervals for the four values of $\mathbb{P}(Y_{2}^{z_Y,z_D,\cor}=1)$ were computed based on 500 bootstrap samples for each of the simulated trials, see \cite[Ch. 5]{Davison_Hinkley_1997} for details.

Firstly, we considered a setting where all postulated models were correctly specified. Then, we know that all estimators are consistent, see Section \ref{sec:Estimation} and Appendix \ref{sec:AppEstimation}.  The simulations confirmed coverage of the true counterfactual probability of interest in the four arms across all considered sample sizes, see  Table \ref{tab:SimusCorrect}.

\begin{table}[ht]
    \centering
    \begin{tabular}{crcccc}
        \hline
        $\bm{(z_Y, z_D)}$ && $\bm{(1,1)}$ & $\bm{(1,0)}$ & $\bm{(0,1)}$ & $\bm{(0,0)}$ \\ \hline
        \multirow{3}{*}{Simple} & n=1000 & 0.948 & 0.948 & 0.958 & 0.970 \\ 
                             & n=5000 & 0.950 & 0.956 & 0.954 & 0.950 \\ 
                             & n=10000 & 0.948 & 0.962 & 0.944 & 0.944 \\ \hdashline
        \multirow{3}{*}{Weighted} & n=1000 & 0.966 & 0.962 & 0.964 & 0.962 \\ 
                                & n=5000 & 0.946 & 0.936 & 0.952 & 0.960 \\ 
                                & n=10000 & 0.958 & 0.964 & 0.936 & 0.952 \\ \hdashline
        \multirow{3}{*}{Doubly robust} & n=1000 & 0.954 & 0.952 & 0.962 & 0.970  \\ 
                                       & n=5000 & 0.950 & 0.952 & 0.956 & 0.948 \\ 
                                      & n=10000 & 0.948 & 0.962 & 0.936 & 0.944 \\ 
    \end{tabular}
    \caption{Fraction of the 500 bootstrapped 95\% confidence intervals which contain the true value of $\mathbb{P}(Y_{2}^{z_Y,z_D,\cor}=1)$ in the scenario where all postulated models are correctly specified.}
    \label{tab:SimusCorrect}
\end{table}

Secondly, we considered a setting where all models were correctly specified, except for the $L_1$ model. Specifically, we fit the model $$\Tilde{\mathbb{P}}_n(L_1=l_1\mid D_1=Y_1=0,L_0,\phi_1(Z))=\widehat{\mathbb{E}}_n[I(L_1=l_1)\mid C_1=D_1=Y_1=0 ], $$ which is miss-specified as it does not stratify for $L_0$ and $Z$. In this setting we do not have consistency for the simple estimator, but the one-step estimator is still consistent (Theorem \ref{thm:DR}). Also, the weighted estimator is identical to the first scenario, as it does not require modelling of $L_1$. The simple estimator has low coverage across all three sample sizes, whereas the influence-function-based estimator has nominal coverage,  see Table \ref{tab:SimusMiss}.

\begin{table}[ht]
    \centering
    \begin{tabular}{crcccc}
        \hline
        $\bm{(z_Y, z_D)}$ && $\bm{(1,1)}$ & $\bm{(1,0)}$ & $\bm{(0,1)}$ & $\bm{(0,0)}$ \\ \hline
        \multirow{3}{*}{Simple} & n=1000 & 0.032 & 0.020 & 0.060 & 0.008 \\ 
                             & n=5000 & 0 & 0 & 0 & 0 \\ 
                             & n=10000 &  0 & 0 & 0 & 0  \\ \hdashline
        \multirow{3}{*}{Doubly robust} & n=1000 & 0.950  & 0.952  & 0.960  & 0.966   \\ 
                                       & n=5000 & 0.950 & 0.956 & 0.948 & 0.950  \\ 
                                       & n=10000& 0.944 & 0.962 & 0.948 & 0.946 \\ 
    \end{tabular}
    \caption{Fraction of the 500 bootstrapped 95\% confidence intervals which contain the true value of $\mathbb{P}(Y_{2}^{z_Y,z_D,\cor}=1)$ in the scenario where all postulated models are correctly specified except for the $L_1$-model, which is not.}
    \label{tab:SimusMiss}
\end{table}

\section{Data application: Acute kidney injury in the SPRINT cohort} \label{sec:SprintExample} 
\subsection{Analysis}
We analysed data from the Systolic Blood Pressure Intervention Trial (SPRINT) \cite{sprint2015randomized}, which randomly assigned individuals to an intensive ($Z=1$) or standard ($Z=0$) blood pressure treatment. Our analysis was based on data for the initial $K+1=30$ months after randomization, and also the following vector of baseline covariates ($L_0$): gender, smoking status, history of clinical or subclinical chronic kidney disease  and  log-mean arterial blood pressure\footnote{In \cite{stensrud2021generalized} mean arterial blood pressure is defined as a third of the Systolic BP plus two thirds of the Diastolic BP. Here we applied the log-transformation to the mean arterial pressure following \cite{BloodPressureJustification} and considered the most-recently-measured value for each time-point.}. Our event of interest is AKI, such that death from any cause is a competing event. After randomization, blood pressure measurements are scheduled for the first three months and every three months afterwards, which we take as a time-varying covariate $L_k$.  We restricted the analysis to individuals with complete baseline covariates, event-free at randomization, and over the age of 75 at baseline. This age restriction allows us to consider the group where most deaths occur \cite{williamson2016intensive}. 

The resulting data set included 1298 individuals in intensive treatment and 1288 in the standard arm. During the 30-month follow-up period 199 individuals were lost to follow-up (censored) in the standard treatment arm, and 205 in the intensive regime. Regarding adherence, we consider  $R_k=1$ when the individual said they 100\% followed their assigned treatment at the clinical visit at time $k+1$. Adherence was initially high and similar in both arms, but it decreased during the observation period, especially in the intensive therapy arm. Further details on adherence in the SPRINT  can be seen in Appendix \ref{sec:AppSprintAdherence}.

Our goal is  to estimate the counterfactual probability of observing the event of interest under interventions on $Z_Y, Z_D$, perfect adherence and no loss of follow-up: $\mathbb{P}(Y_{30}^{z_Y,z_D,\cor}=1)$. Following \cite[Sec. 7.2]{stensrud2021generalized}, we  assume that the dismissible component conditions (Assumption \ref{ass:DCC}) hold under a partition such that $L_k=L_{D,k}$ and $L_{Y,k}=\emptyset$.  Then, we estimated  AKI probabilities using the weighted estimator $\widehat{\nu}_{weighted,Y}$, as $W_{L_Y,s}=1$. Details on the parametric models which we fit to compute this estimator can be seen in Appendix \ref{sec:AppSprintModels}.

\subsection{Results}\label{sec:SprintResults}
The estimates of $\mathbb{P}(Y_{30}^{z_Y,z_D,\cor}=1)$, for the four possible combinations of $z_Y,z_D\in\{0,1\}$, suggest increased risk of AKI under any treatment strategy involving a component of the intensive therapy with respect to the standard treatment (Table \ref{tab:EffectEstimates} and Figure \ref{fig:CumIncWeighted}). However, the confidence intervals are overlapping (Table  \ref{tab:EffectEstimates}). Assuming that the dismissible component conditions hold under $L_Y=\emptyset$, then Proposition \ref{prop:DccImpliesAYPI} implies that $Z_Y$ partial isolation holds. Therefore, the $Z_Y$ sustained separable effect becomes of interest, as it quantifies only direct effects which do not involve the competing event. Particularly, this effect evaluated at $z_D=1$ reflects the direct effect of the $Z_Y$ component of the treatment on the risk of AKI.  The estimate of this sustained effect is -0.102 (95\% confidence interval $[-0.322, 0.412]$), see Table \ref{tab:EffectEstimates}).  Thus,  we find no clear evidence that  removing the $Z_Y$ component of the intensive blood pressure therapy will change the risk of AKI at 30 months.

\begin{table}[ht]
\def~{\hphantom{0}}
\begin{tabular}{lrc}
\hline
{$\bm{(z_Y,z_D)}$} & \textbf{Estimate} & \textbf{95\% CI} \\
\hline
(1,1) & 0.206 & [0.096, 0.361] \\ 
(1,0) & 0.207 & [0.049, 0.412] \\ 
(0,1) & 0.308 & [0.049, 0.608] \\ 
(0,0) & 0.096 & [0.060, 0.163] \\ \hline
\textbf{Causal effect} & -0.102 & [-0.322, 0.412] \\
\end{tabular}

\caption{Estimates of $\mathbb{P}(Y_{30}^{z_Y,z_D,\cor}=1)$ using the weighted estimator $\widehat{\nu}_{weighted,Y}$, and of the causal effect of interest: The $Z_Y$ sustained separable effect evaluated at $z_D=1$ in the additive scale ($\mathbb{P}(Y_{30}^{z_Y=1,z_D=1,\cor}=1)-\mathbb{P}(Y_{30}^{z_Y=0,z_D=1,\cor}=1)$). The quantile-based confidence intervals  are derived from 500 bootstrap samples.} 
\label{tab:EffectEstimates}
\end{table}
\begin{figure}
    \centering
    \includegraphics[width=0.75\linewidth]{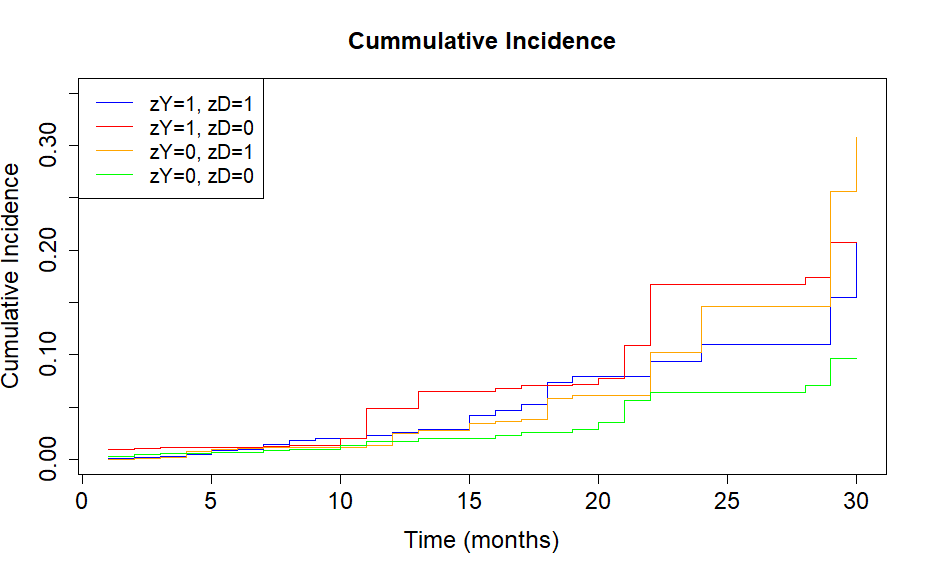}
    \caption{Estimated counterfactual probability of observing an AKI event $\mathbb{P}(Y_{k}^{z_Y,z_D,\cor}=1)$ in the observation time $k\in\{1,\ldots,30\}$ for $z_Y,z_D\in\{0,1\}$, estimated using  the weighted estimator $\widehat{\nu}_{weighted,Y}$. }
    \label{fig:CumIncWeighted}
\end{figure}

\section{Discussion}\label{sec:Discussion}

We have generalized intention-to-treat analyses discussed in \cite{stensrud2021generalized, stensrud2023conditional, wanis2024separable} and, to our knowledge, given the first methods for sustained separable effects. The use of the strategy-centered encoding simplified the exposition of the identification conditions as well as the representation of mechanistic assumptions. 

Furthermore, this work relates to other areas of statistics and causal inference, where it can guide future methodological developments. Beyond a competing event setting, arguments similar to ours can justify time-varying interventionist mediation analysis, broadly defined, see   Appendix \ref{sec:ConnectionInterventionistMediation}. The one-step estimator in Section \ref{sec:DRMainBody} can also be adapted to construct doubly robust estimators for targets considered in related work, see Appendix \ref{sec:AppRelatedWorkDR}.

Finally, the strategy-centered encoding can be used to estimate other sustained treatment effects in observational studies \cite{hernan2017per,hernan2006comparison,cain2010start}, not just for mechanistic estimands like separable effects. In fact, we conjecture that there exist other situations where the encoding helps make identification arguments clearer or easier, which could support the discovery and explanation of new results.

\section*{Acknowledgement}
Ignacio Gonz\'alez-P\'erez and Mats J. Stensrud  were supported by the Swiss National Science Foundation, grant 211550.

\clearpage
\bibliography{references}
\bibliographystyle{unsrt}

\clearpage
\appendix
\section{Proofs of the results in the main text}
\subsection{Relation between Assumption \ref{ass:DCC} and partial isolation} \label{sec:AppProofDCCimpliesPI}
\propDccImpliesAYPI*
\begin{proof}
    This proof closely follows the arguments of the proof of \cite[Lemma 5, App. C]{stensrud2021generalized}. Assume that Assumption \ref{ass:DCC} holds under $\overline{L}_{Y,K}=\emptyset$, but that $Z_Y$ partial isolation does not hold. As per Definition \ref{def:AYPI}, there is a directed path from $Z_Y$ to $D_{k+1}$ (for a certain $k$) in the extended causal DAG which does not intersect $\overline{Y}_k, \overline{R}_{k+1}, \overline{C}_{k+1}$. W.l.o.g. assume this path does not intersect $\overline{D}_{k}$. Given the structure of the extended causal DAG, there are two possibilities.  If this path has no intermediate nodes: $Z_Y\to D_{K+1}$, then the second condition in Equation \eqref{eq:DCC} would be violated, which is a contradiction. If this path has intermediate nodes, then all the nodes are either in   $\overline{L}_K=\overline{L}_{D,K}$ or are unmeasured. If all nodes are unmeasured then second condition in Equation \eqref{eq:DCC} would again be violated. Otherwise, the fourth dismissible component condition would be violated with respect to the first node of $\overline{L}_{D,K}$ which is encountered in this path. This concludes the proof.
\end{proof}
\begin{restatable}{proposition}{propDccImpliesADPI}\label{prop:DccImpliesADPI}
    If the dismissible component conditions (Assumption \ref{ass:DCC}) hold under a partition where $\overline{L}_{D,K}=\emptyset$, then $Z_D$ partial isolation holds.
\end{restatable}
\begin{proof}
    The proof follows analogously to that of Proposition \ref{prop:DccImpliesAYPI}.
\end{proof}

\begin{restatable}{corollary}{coroDccImpliesFullPI}\label{coro:DccImpliesFullPI}
    If the dismissible component conditions (Assumption \ref{ass:DCC}) hold under both partition $\overline{L}_{Y,K}=\emptyset,\overline{L}_{D,K}=\overline{L}_{K} $ and $\overline{L}_{D,K}=\emptyset,\overline{L}_{Y,K}=\overline{L}_{K} $ then full  isolation holds.
\end{restatable}
\begin{proof}
    Follows directly from Propositions \ref{prop:DccImpliesAYPI} and \ref{prop:DccImpliesADPI} and the definition of full isolation.
\end{proof}

\subsection{Proof of Theorem \ref{thm:IdentifcationFormula}}\label{sec:AppProofIdentification}
We follow identical steps to the proof of the identification formula presented in \cite[App. B]{stensrud2021generalized}. We begin by proving two lemmas which we will use to give a proof of Theorem \ref{thm:IdentifcationFormula}.
\begin{lemma}\label{lemma:Identification1}
    Under a FFRCISTG model, Assumption \ref{ass:DCC} implies the following equalities for $z_Y,z_D\in\{0,1\}$:
    \begin{align*}
        &\mathbb{P}(Y_{s+1}^{z_Y,z_D,\cor}=1\mid D_{s+1}^{z_Y,z_D,\cor}=Y_{s}^{z_Y,z_D,\cor}=0,\overline{L}_s^{z_Y,z_D,\cor}=\overline{l}_s)\\
        &=\mathbb{P}(Y_{s+1}^{z_Y,\cor}=1\mid D_{s+1}^{z_Y,\cor}=Y_{s}^{z_Y,\cor}=0,\overline{L}_s^{z_Y,\cor}=\overline{l}_s),\\
        &\mathbb{P}(D_{s+1}^{z_Y,z_D,\cor}=1\mid D_{s}^{z_Y,z_D,\cor}=Y_{s}^{z_Y,z_D,\cor}=0,\overline{L}_s^{z_Y,z_D,\cor}=\overline{l}_s)\\
        &=\mathbb{P}(D_{s+1}^{z_D,\cor}=1\mid D_{s}^{z_D,\cor}=Y_{s}^{z_D,\cor}=0,\overline{L}_s^{z_D,\cor}=\overline{l}_s),\\
        &\mathbb{P}(L_{D,s}^{z_Y,z_D,\cor}=l_{D,s}\mid D_{s}^{z_Y,z_D,\cor}=Y_{s}^{z_Y,z_D,\cor}=0,\overline{L}_{s-1}^{z_Y,z_D,\cor}=\overline{l}_{s-1})\\
        &=\mathbb{P}(L_{D,s}^{z_D,\cor}=l_{D,s}\mid D_{s}^{z_D,\cor}=Y_{s}^{z_D,\cor}=0,\overline{L}_{s-1}^{z_D,\cor}=\overline{l}_{s-1}),
    \end{align*} 
    \resizebox{\linewidth}{!}{$
        \begin{aligned}
        \mathbb{P}&(L_{Y,s}^{z_Y,z_D,\cor}=l_{Y,s}\mid D_{s}^{z_Y,z_D,\cor}=Y_{s}^{z_Y,z_D,\cor}=0,{L}_{D,s}^{z_Y,z_D,\cor}={l}_{D,s},\overline{L}_{s-1}^{z_Y,z_D,\cor}=\overline{l}_{s-1})\\
        =&\mathbb{P}(L_{Y,s}^{z_Y,\cor}=l_{Y,s}\mid D_{s}^{z_Y,\cor}=Y_{s}^{z_Y,\cor}=0,{L}_{D,s}^{z_Y,\cor}={l}_{D,s},\overline{L}_{s-1}^{z_Y,\cor}=\overline{l}_{s-1}),
    \end{aligned}
    $}\vspace{1em}
    where recall that the counterfactual notation $Y^z$ means $Y^{Z=z}$ which is by the generalized decomposition assumption equivalent to $Y^{Z_Y=Z_D=z}$.
\end{lemma}
\begin{proof}
    The proof of this Lemma is identical to the proof of \cite[Lemma 1, App. B]{stensrud2021generalized}, where the interventions now include $\cor$ instead of just $\overline{c}=0$.
\end{proof}
\begin{lemma}\label{lemma:Identification2}
    Under Assumption \ref{ass:Exchangeability}-\ref{ass:Positivity}, for $s=0,\ldots,K$ and $a\in\{0,1\}$ we have that:
    \begin{align*}
        &\mathbb{P}(Y_{s+1}^{z,\cor}=1\mid D_{s+1}^{z,\cor}=Y_{s}^{z,\cor}=0,\overline{L}_s^{z,\cor}=\overline{l}_s)\\&=\mathbb{P}(Y_{s+1}=1\mid C_{s+1}=D_{s+1}=Y_s=0,\overline{L}_s=\overline{l}_s,Z=z,\overline{R}_{s+1}=1),\\
        &\mathbb{P}(D_{s+1}^{z,\cor}=1\mid D_{s}^{z,\cor}=Y_{s}^{z,\cor}=0,\overline{L}_s^{z,\cor}=\overline{l}_s)\\&=\mathbb{P}(D_{s+1}=1\mid C_{s+1}=D_{s}=Y_s=0,\overline{L}_s=\overline{l}_s,Z=z,\overline{R}_{s+1}=1),\\
        &\mathbb{P}(L_{Y,s}^{z,\cor}=l_{Y,s}\mid D_{s}^{z,\cor}=Y_{s}^{z,\cor}=0,\overline{L}_{s-1}^{z,\cor}=\overline{l}_{s-1},{L}_{D,s}^{z,\cor}={l}_{D,s})\\&=\mathbb{P}(L_{Y,s}=l_{Y,s}\mid C_{s}=D_{s}=Y_s=0,\overline{L}_{s-1}=\overline{l}_{s-1},{L}_{D,s}={l}_{D,s},Z=z,\overline{R}_{s}=1),\\
        &\mathbb{P}(L_{D,s}^{z,\cor}=l_{D,s}\mid D_{s}^{z,\cor}=Y_{s}^{z,\cor}=0,\overline{L}_{s-1}^{z,\cor}=\overline{l}_{s-1})\\&=\mathbb{P}(L_{D,s}=l_{D,s}\mid C_{s}=D_{s}=Y_s=0,\overline{L}_{s-1}=\overline{l}_{s-1},Z=z,\overline{R}_{s}=1).
    \end{align*}
\end{lemma}
\begin{proof}
    We present here the proof of the first equality, as the rest come analogously. This proof follows the steps of the proof of \cite[Lemma 2, App. B]{stensrud2021generalized}. Now, using the laws of probability, exchangeability, consistency for $L_0$, positivity, and the fact that all subjects are uncensored and event free at time $s=0$ and $R_0=\emptyset$ we see that: 
    \begin{center}
        \rotatebox{90}{$
    \begin{aligned}
        \mathbb{P}&(Y_{s+1}^{z,\cor}=1\mid D_{s+1}^{z,\cor}=Y_{s}^{z,\cor}=0,\overline{L}_s^{z,\cor}=\overline{l}_s)\\
        =&\mathbb{P}(Y_{s+1}^{z,\cor}=1\mid D_{s+1}^{z,\cor}=Y_{s}^{z,\cor}=0,\overline{L}_s^{z,\cor}=\overline{l}_s,Y_0=D_0=C_0=0,L_0=l_0)\\
        =&\frac{\mathbb{P}(Y_{s+1}^{z,\cor}=1,D_{s+1}^{z,\cor}=Y_{s}^{z,\cor}=0,\overline{L}_s^{z,\cor}=\overline{l}_s\mid Y_0=D_0=C_0=0,L_0=l_0)}{\mathbb{P}(D_{s+1}^{z,\cor}=Y_{s}^{z,\cor}=0,\overline{L}_s^{z,\cor}=\overline{l}_s\mid Y_0=D_0=C_0=0,L_0=l_0)}\\
        =&\frac{\mathbb{P}(Y_{s+1}^{z,\cor}=1,D_{s+1}^{z,\cor}=Y_{s}^{z,\cor}=0,\overline{L}_s^{z,\cor}=\overline{l}_s\mid Y_0=D_0={C}_1=0,L_0=l_0,Z=z,\overline{R}_1=1)}{\mathbb{P}(D_{s+1}^{z,\cor}=Y_{s}^{z,\cor}=0,\overline{L}_s^{z,\cor}=\overline{l}_s\mid Y_0=D_0={C}_1=0,L_0=l_0,Z=z,\overline{R}_1=1)}\\
        =&\mathbb{P}(Y_{s+1}^{z,\cor}=1\mid D_{s+1}^{z,\cor}=Y_{s}^{z,\cor}=0,\overline{L}_s^{z,\cor}=\overline{l}_s,L_0=l_0,Z=z,\overline{R}_1=1,C_1=0 )
    \end{aligned}
    $}
    \end{center}
    The case $s=0$ follows from consistency and the fact that $Y_0=C_0=D_0=0$ by design:
    \begin{align*}
        \mathbb{P}(Y_{1}^{z,\cor}=1\mid& D_{1}^{z,\cor}=0,L_0=l_0,Z=z,{C}_1=0,\overline{R}_1=1)\\
        =&\mathbb{P}(Y_{1}=1\mid D_{1}=0,L_0=l_0,Z=z,{C}_1=0,\overline{R}_1=1).
    \end{align*}
    Now we study the case $s\geq 1$. We have that:
    \begin{center}
        \rotatebox{90}{$
        \begin{aligned}
        \mathbb{P}&(Y_{s+1}^{z,\cor}=1\mid D_{s+1}^{z,\cor}=Y_{s}^{z,\cor}=0,\overline{L}_s^{z,\cor}=\overline{l}_s,L_0=l_0,Z=z,{C}_1=0,\overline{R}_1=1)\\
        \stackrel{(*)}{=}&\mathbb{P}(Y_{s+1}^{z,\cor}=1\mid D_{s+1}^{z,\cor}=Y_{s}^{z,\cor}=0=Y_1=D_1,\overline{L}_s^{z,\cor}=\overline{l}_s,\overline{L}_1=\overline{l}_1,Z=z,{C}_1=0,\overline{R}_1=1)\\
        =&\frac{\mathbb{P}(Y_{s+1}^{z,\cor}=1 ,D_{s+1}^{z,\cor}=Y_{s}^{z,\cor}=0,\overline{L}_s^{z,\cor}=\overline{l}_s\mid Y_1=D_1=0,\overline{L}_1=\overline{l}_1,Z=z,{C}_1=0,\overline{R}_1=1)}{\mathbb{P}(D_{s+1}^{z,\cor}=Y_{s}^{z,\cor}=0,\overline{L}_s^{z,\cor}=\overline{l}_s\mid Y_1=D_1=0,\overline{L}_1=\overline{l}_1,Z=z,{C}_1=0,\overline{R}_1=1)}\\
        \stackrel{(**)}{=}&\frac{\mathbb{P}(Y_{s+1}^{z,\cor}=1 ,D_{s+1}^{z,\cor}=Y_{s}^{z,\cor}=0,\overline{L}_s^{z,\cor}=\overline{l}_s\mid Y_1=D_1=0,\overline{L}_1=\overline{l}_1,Z=z,{C}_2=0,\overline{R}_2=1)}{\mathbb{P}(D_{s+1}^{z,\cor}=Y_{s}^{z,\cor}=0,\overline{L}_s^{z,\cor}=\overline{l}_s\mid Y_1=D_1=0,\overline{L}_1=\overline{l}_1,Z=z,{C}_2=0,\overline{R}_2=1)}\\
        =&\mathbb{P}(Y_{s+1}^{z,\cor}=1\mid D_{s+1}^{z,\cor}=Y_{s}^{z,\cor}=0=Y_1=D_1,\overline{L}_s^{z,\cor}=\overline{l}_s,\overline{L}_1=\overline{l}_1,Z=z,{C}_2=0,\overline{R}_2=1)\\
        \stackrel{(***)}{=}&\mathbb{P}(Y_{s+1}^{z,\cor}=1\mid D_{s+1}^{z,\cor}=Y_{s}^{z,\cor}=0=Y_2=D_2,\overline{L}_s^{z,\cor}=\overline{l}_s,\overline{L}_2=\overline{l}_2,Z=z,{C}_2=0,\overline{R}_2=1), 
    \end{aligned}
        $}
    \end{center}
    which follows from:
    \begin{itemize}
        \item[$(*)$] Consistency for the case $k=0$ and positivity,
        \item[$(**)$] Exchangeability for $k=1$ and positivity,
        \item[$(***)$] Consistency for $k=1$.
    \end{itemize}
    Iterating these previous steps we conclude that:
    \begin{align*}
        \mathbb{P}(Y_{s+1}^{z,\cor}=1\mid& D_{s+1}^{z,\cor}=Y_{s}^{z,\cor}=0,\overline{L}_s^{z,\cor}=\overline{l}_s,L_0=l_0,Z=z,\\&{C}_1=0,\overline{R}_1=1)\\
        {=}&\mathbb{P}(Y_{s+1}=1\mid D_{s+1}=Y_{s}=0,\overline{L}_s=\overline{l}_s,Z=z,{C}_{s+1}=0,\overline{R}_{s+1}=1), 
    \end{align*}
    as we wanted to show.
\end{proof}
\thmIdentifcationFormula*
\begin{proof}
    Lemma \ref{lemma:Identification1} allows us to change interventional quantities in the four-arm trial to  the two-arm trial, and Lemma \ref{lemma:Identification2} shows  how to change these counterfactual quantities to factual ones. The proof is then straightforward after applying the law of total probability and taking into account all the disjoint cases in which we can observe the target event. Consider the case where $z_Y\neq z_D$. Then:
    \resizebox{\linewidth}{!}{$
        \begin{aligned}
        &\mathbb{P}(Y_{K+1}^{z_Y,z_D,\cor}=1)\\=&\sum_{\overline{l}_K}\mathbb{P}(Y_{K+1}^{z_Y,z_D,\cor}=1,\overline{L}_K^{z_Y,z_D,\cor}=\overline{l}_K)\\
        =&\sum_{j=0}^{K}\sum_{\overline{l}_K}  \mathbb{P}  ( Y_{j+1}^{z_Y,z_D,\cor}=1 \mid D_{j+1}^{z_Y,z_D,\cor}=Y_j^{z_Y,z_D,\cor}=0,  \overline{L}^{z_Y,z_D,\cor}_{j}=\overline{l}_{j})  \\
        &  \prod_{s=0}^{j} \bigg[ \mathbb{P} ( D_{s+1}^{z_Y,z_D,\cor}=0 \mid Y_s^{z_Y,z_D,\cor}=D_{s}^{z_Y,z_D,\cor}=0, \overline{L}_{s}^{z_Y,z_D,\cor} = \overline{l}_{s})  \\
        & \times \mathbb{P} ( L_{Y,s}^{z_Y,z_D,\cor}=l_{Y,s} \mid Y_s^{z_Y,z_D,\cor}=D_{s}^{z_Y,z_D,\cor}=0, L_{D,s}^{z_Y,z_D,\cor} =  l_{D,s}, \overline{L}_{s-1}^{z_Y,z_D,\cor} = \overline{l}_{s-1})  \\
        & \times \mathbb{P} ( L_{D,s}^{z_Y,z_D,\cor}=l_{D,s} \mid Y_s^{z_Y,z_D,\cor}=D_{s}^{z_Y,z_D,\cor}=0, \overline{L}_{s-1}^{z_Y,z_D,\cor} = \overline{l}_{s-1})\\
        & \times \mathbb{P}  ( Y_{s}^{z_Y,z_D,\cor}=0 \mid D_{s}^{z_Y,z_D,\cor}=Y_{s-1}^{z_Y,z_D,\cor}=0,  \overline{L}_{s-1}^{z_Y,z_D,\cor}=\overline{l}_{s-1}) \bigg]\\
        \stackrel{\text{Lemma } \ref{lemma:Identification1}}{=}&\sum_{j=0}^{K}\sum_{\overline{l}_K}  \mathbb{P}  ( Y_{j+1}^{z_Y,\cor}=1 \mid D_{j+1}^{z_Y,\cor}=Y_j^{z_Y,\cor}=0,  \overline{L}^{z_Y,\cor}_{j}=\overline{l}_{j})  \\
        &  \prod_{s=0}^{j} \bigg[ \mathbb{P} ( D_{s+1}^{z_D,\cor}=0 \mid Y_s^{z_D,\cor}=D_{s}^{z_D,\cor}=0, \overline{L}_{s}^{z_D,\cor} = \overline{l}_{s})  \\
        & \times \mathbb{P} ( L_{Y,s}^{z_Y,\cor}=l_{Y,s} \mid Y_s^{z_Y,\cor}=D_{s}^{z_Y,\cor}=0, L_{D,s}^{z_Y,\cor} =  l_{D,s}, \overline{L}_{s-1}^{z_Y,\cor} = \overline{l}_{s-1})  \\
        & \times \mathbb{P} ( L_{D,s}^{z_D,\cor}=l_{D,s} \mid Y_s^{z_D,\cor}=D_{s}^{z_D,\cor}=0, \overline{L}_{s-1}^{z_D,\cor} = \overline{l}_{s-1})\\
        & \times \mathbb{P}  ( Y_{s}^{z_Y,\cor}=0 \mid D_{s}^{z_Y,\cor}=Y_{s-1}^{z_Y,\cor}=0,  \overline{L}_{s-1}^{z_Y,\cor}=\overline{l}_{s-1}) \bigg]\\
        \stackrel{\text{Lemma } \ref{lemma:Identification2}}{=}&\sum_{j=0}^{K}\sum_{\overline{l}_K}  \mathbb{P}  ( Y_{j+1}=1 \mid D_{j+1}=Y_j=C_{j+1}=0,  \overline{L}_{j}=\overline{l}_{j}, Z=z_Y, \overline{R}_{j+1}=1)  \\
        &  \prod_{s=0}^{j} \bigg[ \mathbb{P} ( D_{s+1}=0 \mid Y_s=D_{s}=C_{s+1}=0, \overline{L}_{s} = \overline{l}_{s}, Z=z_D, \overline{R}_{s+1}=1)  \\
        & \times \mathbb{P} ( L_{Y,s}=l_{Y,s} \mid Y_s=D_{s}=C_{s}=0, L_{D,s} =  l_{D,s}, \overline{L}_{s-1} = \overline{l}_{s-1}, Z =z_Y, \overline{R}_{s}=1)  \\
        & \times \mathbb{P} ( L_{D,s}=l_{D,s} \mid Y_s=D_{s}=C_{s}=0, \overline{L}_{s-1} = \overline{l}_{s-1}, Z=z_D, \overline{R}_{s}=1)\\
        & \times \mathbb{P}  ( Y_{s}=0 \mid D_{s}=Y_{s-1}=C_{s}=0,  \overline{L}_{s-1}=\overline{l}_{s-1}, Z=z_Y, \overline{R}_{s}=1) \bigg].
    \end{aligned}
    $}\vspace{1em}\\
    The proof in the case $z_Y=z_D$ is follows the same steps as the one just presented, but now Lemma \ref{lemma:Identification1} is not needed as the third equality follows from the generalized decomposition assumption. 
\end{proof}

\subsection{Proofs of the results in Section \ref{sec:Estimation}}\label{sec:AppProofEstimation}
\thmWeigthedFormulas*
\begin{proof}
 Firstly consider     $$W^*_{(C,R),s} (z) = \frac{1 }{ \prod\limits_{j=0}^{s}  \mathbb{P}(C_{j+1}=0, R_{j+1}=1 \mid C_{j}=D_{j}=Y_j=0, \overline{L}_{j},  Z=z, \overline{R}_{j}=1) },$$ so that $ {W}_{(C,R),s}(z) = I(C_{s+1} =0)I( \overline{R}_{s+1}=1)W^*_{(C,R),s}(z)$. For $s=0,\ldots,K$ we have that: 
    \begin{align*}
        \mathbb{E} &[ W_{(C,R),s}(z_Y) W_{D,s}  W_{L_{D},s}  (1-Y_s)(1-D_{s+1}) Y_{s+1} \mid Z=z_Y]\\
        =&\mathbb{E}[ W^*_{(C,R),s}(z_Y) W_{D,s}  W_{L_{D},s} (1-C_{s+1})\overline{R}_{s+1} (1-Y_s)(1-D_{s+1}) Y_{s+1} \mid Z=z_Y]\\
        =&\sum_{\overline{l}_K}\mathbb{P}(Y_{s+1}=1, Y_s=D_{s+1}=C_{s+1}=0, \overline{R}_{s+1}=1, \overline{L}_K=\overline{l}_K\mid Z=z_Y)\times\\& \qquad W^*_{(C,R),s}(z_Y) W_{D,s}  W_{L_{D},s}\\
        =&\sum_{\overline{l}_K}\bigg[\mathbb{P}(Y_{s+1}=1\mid Y_s=D_{s+1}=C_{s+1}=0, \overline{R}_{s+1}=1, \overline{L}_K=\overline{l}_K, Z=z_Y)\times\\
        &\prod_{j=0}^s \left\{\mathbb{P}(D_{j+1}=0\mid Y_j=D_{j}=C_{j+1}=0, \overline{R}_{j+1}=1, \overline{L}_j=\overline{l}_j, Z=z_Y)\times\right.\\
        &\mathbb{P}(C_{j+1}=0, R_{j+1}=1\mid Y_j=D_{j}=C_{j}=0, \overline{R}_{j}=1, \overline{L}_j=\overline{l}_j, Z=z_Y)\times\\
        &\mathbb{P}(L_{Y,j}=l_{Y,j}\mid Y_j=D_{j}=C_{j}=0, \overline{R}_{j}=1, L_{D,j}=l_{D,j},\overline{L}_{j-1}=\overline{l}_{j-1}, Z=z_Y)\times\\
        &\mathbb{P}(L_{D,j}=l_{D,j}\mid Y_j=D_{j}=C_{j}=0, \overline{R}_{j}=1, \overline{L}_{j-1}=\overline{l}_{j-1}, Z=z_Y)\times\\
        &\left.\mathbb{P}(Y_{j}=0\mid Y_{j-1}=D_{j}=C_{j}=0, \overline{R}_{j}=1, \overline{L}_{j-1}=\overline{l}_{j-1}, Z=z_Y)\right\}\times\\&W^*_{(C,R),s}(z_Y) W_{D,s}  W_{L_{D},s}\bigg]\\
        =&\sum_{\overline{l}_K}\mathbb{P}(Y_{s+1}=1\mid Y_s=D_{s+1}=C_{s+1}=0, \overline{R}_{s+1}=1, \overline{L}_K=\overline{l}_K, Z=z_Y)\times\\
        &\prod_{j=0}^s \left\{\mathbb{P}(D_{j+1}=0\mid Y_j=D_{j}=C_{j+1}=0, \overline{R}_{j+1}=1, \overline{L}_j=\overline{l}_j, Z=z_D)\times\right.\\
        &\mathbb{P}(L_{Y,j}=l_{Y,j}\mid Y_j=D_{j}=C_{j}=0, \overline{R}_{j}=1, L_{D,j}=l_{D,j},\overline{L}_{j-1}=\overline{l}_{j-1}, Z=z_Y)\times\\
        &\mathbb{P}(L_{D,j}=l_{D,j}\mid Y_j=D_{j}=C_{j}=0, \overline{R}_{j}=1, \overline{L}_{j-1}=\overline{l}_{j-1}, Z=z_D)\times\\
        &\left.\mathbb{P}(Y_{j}=0\mid Y_{j-1}=D_{j}=C_{j}=0, \overline{R}_{j}=1, \overline{L}_{j-1}=\overline{l}_{j-1}, Z=z_Y)\right\}.
    \end{align*}
    If we now sum this last expression over $s=0,\ldots,K$, we obtain precisely the identification formula \eqref{eq:IdentificationFormula}. This concludes the proof.
\end{proof}

\section{Constructing a DAG under the strategy-centered encoding}\label{sec:AppAlgoDrawDags}
As we mentioned in Section \ref{sec:AdherenceEncodings}, we take the treatment-centered encoding as primitive. Given a DAG $\mathcal{G}$ under this encoding, we have introduced a procedure to construct a corresponding DAG $\mathcal{G}^*$ under the strategy-centered encoding, reflecting the definitional relation between both encodings. We present this procedure in full generality in Algorithm \ref{alg:DagDrawing}. We can also see from this algorithm what we anticipated in the main text: $\mathcal{G}^*$ has fewer nodes than $\mathcal{G}$ when $k>2$, but more edges, except in trivial cases where treatments have no effects. Nevertheless, the identification conditions can be expressed more concisely in $\mathcal{G}^*$, as it will become evident in Appendix \ref{sec:EquivalenceAdherenceEncodings}.

\begin{algorithm}
\caption{Algorithm to construct a DAG $\mathcal{G}^*$ under the strategy-centered encoding based on a DAG $\mathcal{G}$ under the primitive, treatment-centered, encoding $V$.}
\label{alg:DagDrawing}
\KwIn{A extended causal DAG $\mathcal{G}$ under the treatment-centered encoding}
Define $\Gamma\colon V\to V^*$ as $\Gamma(A_{W,k})=R_k$ for $W \in \{D, Y\}$ and $k\in\{1,\ldots,K+1\}$, and as the identity otherwise\; 
Define $\mathcal{G}^*$ as the empty DAG over $V^*$\;
\For{$W \in \{D, Y\}$}{
    \For{$k \gets 1$ \KwTo $K+1$}{
        \ForEach{$X$ a child of $A_{W,k}$ in $\mathcal{G}$}{
            Add an edge $Z_W \to \Gamma(X)$ in $\mathcal{G}^*$\;
            \If{$\Gamma(X)\neq\Gamma(A_{W,k})$}{
                Add an edge $R_k \to \Gamma(X)$ in $\mathcal{G}^*$\;
            }
        }
        \ForEach{$X$ a parent of $A_{W,k}$ in $\mathcal{G}$}{
            \If{the edge $\Gamma(X) \to R_k$ does not exist in $\mathcal{G}^*$}{
                Add an edge $\Gamma(X) \to R_k$ in $\mathcal{G}^*$\;
            }
        }
        \If{the edge $Z_W \to R_k$ does not exist in $\mathcal{G}^*$}{
            Add the edge $Z_W \to R_k$ to $\mathcal{G}^*$\;
        }
    }
}
Add all edges in $\mathcal{G}_{V \backslash \{A_Y, A_D\}}$\footnotemark[1] to $\mathcal{G}^*$\;
\KwOut{DAG $\mathcal{G}^*$}
\end{algorithm}
\footnotetext[1]{For a DAG $\mathcal{G}$ over nodes $V$ and $B\subset V$ we define the sub-DAG $\mathcal{G}_B$ as the DAG over $B$ obtained by retaining all edged in $\mathcal{G}$ between nodes in $B$.}

Furthermore, we argue that Algorithm \ref{alg:DagDrawing} is the appropriate way to construct the extended causal DAG under the strategy-centered encoding. As we will see in Appendix \ref{sec:AppEquivalenceIdentificationConditions}, this procedure yields two causal DAGs which, under a faithfulness assumption, make the dismissible component conditions equivalent under both encodings.

\section{Equivalence of the strategy- and treatment-centered adherence encodings towards identification}\label{sec:EquivalenceAdherenceEncodings}
\subsection{Identification under the treatment-centered encoding}
We have seen in Section \ref{sec:AdherenceEncodings}  how  the perfect adherence intervention can be described equivalently using both encodings, $\{Z=z, \overline{R}_{K+1}=1\}\Leftrightarrow \{\overline{A}_{K+1}=a\}$. Thus,  it could be expected that both allow for suitable conditions under which the quantity of interest can be identified, and that the identification formulas should differ only on an adherence reparametrization. This will be the focus of this section.

Consider then the observed data structure introduced in Section \ref{sec:ObservedDataStructure} but now under the treatment-centered encoding.
As explained in Section \ref{sec:TreatmentDecomposition}, in a separable effects context under this time-varying setting, the treatment-centered encoding requires $A_k$ to be decomposed into ${A_{D,k},A_{Y,k}}$ at every time point following the hypothetical treatment decomposition. In contrast, the strategy-centered encoding requires separation only of $Z$, but not of $R_k$.
In this context, the sustained separable effects (Definitions \ref{def:AYSeparableEffect}-\ref{def:ADSeparableEffect}) are defined analogously but under interventions which set $\{\overline{A}_{Y,K+1}=a_Y,\ \overline{A}_{D,K+1}=a_D,\ C_{K+1}=0\}$, which represents no loss of follow-up and perfect adherence to the treatment strategy $Z_Y=a_Y,\ Z_D=a_D$. To show the desired equivalence between the two considered adherence encodings, we first show that under conditions similar to those required in Theorem \ref{thm:IdentifcationFormula} we can identify $\mathbb{P}(Y_{K+1}^{\overline{A}_{Y,K}=a_Y,\overline{A}_{D,K}=a_D,\overline{c}=0}=1)$ in the two-arm trial. We first begin by adapting the identifiability assumptions from Section \ref{sec:Identification} to the treatment-centered adherence encoding. 
\begin{assumption}[Identifiability conditions under the treatment-centered encoding]\label{ass:EquivalenceIdentifiability}
For $a\in\{0,1\}$ and  $k\in\{0,\ldots,K\}$
\begin{align}\label{eq:ExchangeabilityEquivalence}
    \underline{Y}_{k+1}^{\overline{A}_{K+1}=a,\overline{c}=0}, \underline{D}_{k+1}^{\overline{A}_{K+1}=a,\overline{c}=0},& \underline{L}_{k+1}^{\overline{A}_{K+1}=a,\overline{c}=0}\independent \\&C_{k+1},A_{k+1}\mid Y_k=D_k=0,\overline{L}_k, \overline{C}_k=0, \overline{A}_{k}=a,\nonumber
\end{align}
\begin{align}\label{eq:PositivityEquivalenceDCC}
    &\mathbb{P}(\overline{L}_{k}=\overline{l}_{k},Y_k=D_{k+1}=C_{k+1}=0,A_1=\ldots=A_{k+1})>0\Rightarrow\\\nonumber  &\qquad\mathbb{P}(\overline{A}_{k+1}=a\mid \overline{L}_{k}=\overline{l}_{k},Y_k=D_{k+1}=C_{k+1}=0)>0
\end{align}
\begin{align}\label{eq:PositivityEquivalence}
    &\mathbb{P}(\overline{L}_{k}=\overline{l}_{k}, \overline{A}_{k}=a, Y_k=D_{k}=C_{k}=0)>0\Rightarrow\\&\qquad \mathbb{P}(A_{k+1}=a,C_{k+1}=0|\overline{L}_{k}=\overline{l}_{k}, \overline{A}_{k}=a,Y_k=D_{k}=C_{k}=0)>0,\nonumber
\end{align}
 and if $C_{k+1}=0$ and $\overline{A}_{k+1}=a$ then
\begin{equation}\label{eq:ConsistencyEquivalence}
    \overline{Y}_{k+1}^{\overline{A}_{K+1}=a,\overline{c}=0}=\overline{Y}_{k+1},\ \overline{D}_{k+1}^{\overline{A}_{K+1}=a,\overline{c}=0}=\overline{D}_{k+1} \text{ and } \overline{L}_{k+1}^{\overline{A}_{K+1}=a,\overline{c}=0}=\overline{L}_{k+1}.
\end{equation}
\end{assumption}
Equations \eqref{eq:ExchangeabilityEquivalence}-\eqref{eq:ConsistencyEquivalence} are  the exchangeability, positivity and consistency conditions under the treatment-centered encoding. As with Assumption \ref{ass:Exchangeability}, the conditional exchangeability conditions \eqref{eq:ExchangeabilityEquivalence} could be expanded into $6(K+1)$  conditional independencies which are slightly weaker than the ones we present here.

Remaining are the dismissible component conditions. As we have done with Assumption \ref{ass:EquivalenceIdentifiability}, we could consider the dismissible component conditions under the strategy-centered encoding (Assumption \ref{ass:DCC}) and translate them to the treatment-centered encoding using the definitional relation between these variables. This procedure yields the following assumption.

\begin{assumption}
Let the time varying covariates be expressed as two components: $L_k=(L_{D,k},L_{Y,k})$. Furthermore, let $G$ refer to  the  trial where $A_{Y,1}$ and $A_{D,1}$ are randomly assigned, but the causal structure between variables is otherwise identical to the observed data. We use the notation $X(G)$ to indicate  a variable $X$ in this trial. Then, for all $k \in \{0,\dots, K\}$:
\begin{align}
     Y_{k+1}^{\underline{A}_{Y,2}=A_{Y,1},\underline{A}_{D,2}=A_{D,1},\overline{c}=0}(G) \independent A_{D,1}(G) \mid& \{A_{Y,1}(G), D_{k+1}^{\underline{A}_{Y,2}=A_{Y,1},\underline{A}_{D,2}=A_{D,1},\overline{c}=0} (G)=0, \\ &\quad Y_k^{\underline{A}_{Y,2}=A_{Y,1},\underline{A}_{D,2}=A_{D,1},\overline{c}=0}(G)=0, \nonumber\\ &\quad\overline{L}_k^{\underline{A}_{Y,2}=A_{Y,1},\underline{A}_{D,2}=A_{D,1},\overline{c}=0}(G)\}, \nonumber\\
      D_{k+1}^{\underline{A}_{Y,2}=A_{Y,1},\underline{A}_{D,2}=A_{D,1},\overline{c}=0}(G) \independent A_{Y,1}(G) \mid& \{A_{D,1}(G), D_{k}^{\underline{A}_{Y,2}=A_{Y,1},\underline{A}_{D,2}=A_{D,1},\overline{c}=0} (G)=0, \nonumber\\ &\quad  Y_k^{\underline{A}_{Y,2}=A_{Y,1},\underline{A}_{D,2}=A_{D,1},\overline{c}=0}(G)=0, \nonumber\\ &\quad \overline{L}_k^{\underline{A}_{Y,2}=A_{Y,1},\underline{A}_{D,2}=A_{D,1},\overline{c}=0}(G)\},\nonumber\\
     L_{Y,k}^{\underline{A}_{Y,2}=A_{Y,1},\underline{A}_{D,2}=A_{D,1},\overline{c}=0}(G) \independent A_{D,1}(G) \mid& \{A_{Y,1}(G), D_{k}^{\underline{A}_{Y,2}=A_{Y,1},\underline{A}_{D,2}=A_{D,1},\overline{c}=0} (G)=0,\nonumber\\ &\quad Y_k^{\underline{A}_{Y,2}=A_{Y,1},\underline{A}_{D,2}=A_{D,1},\overline{c}=0}(G)=0, \nonumber\\ & \quad L_{D,k}^{\underline{A}_{Y,2}=A_{Y,1},\underline{A}_{D,2}=A_{D,1},\overline{c}=0}(G),\nonumber\\ & \quad \overline{L}_{k-1}^{\underline{A}_{Y,2}=A_{Y,1},\underline{A}_{D,2}=A_{D,1},\overline{c}=0}(G)\}, \nonumber\\
     L_{D,k}^{\underline{A}_{Y,2}=A_{Y,1},\underline{A}_{D,2}=A_{D,1},\overline{c}=0}(G) \independent \{A_{Y,1}(G) \mid& A_{D,1}(G), D_{k}^{\underline{A}_{Y,2}=A_{Y,1},\underline{A}_{D,2}=A_{D,1},\overline{c}=0} (G)=0, \nonumber\\ &\quad Y_k^{\underline{A}_{Y,2}=A_{Y,1},\underline{A}_{D,2}=A_{D,1},\overline{c}=0}(G)=0, \nonumber\\ &\quad  \overline{L}_{k-1}^{\underline{A}_{Y,2}=A_{Y,1},\underline{A}_{D,2}=A_{D,1},\overline{c}=0}(G)\}.  \nonumber
\end{align}     
\end{assumption}

The main issue with these conditions is that they are read off a SWIG where fixed nodes are random, as their value comes from the realization of a random node. To avoid this potential complication we go back to basic principles, as postulate a different set of dismissible component conditions based on the fact that $\overline{A}_{K+1}$ is split at every time point. This alternative conditions will still prove sufficient for identification.

\begin{assumption}\label{ass:EquivalenceDCC}
    Let the time varying covariates be expressed as two components: $L_k=(L_{Y,k},L_{D,k})$, and let $H$ refer to  the  trial where  $A_Y$ and $A_D$ are observed separately at every time point. We use the notation $X(H)$ to indicate  a variable $X$ in this trial. Then for all $k \in \{0,\dots, K\}$:
\begin{align}\label{eq:DCCsEquivalence}\allowdisplaybreaks
    Y_{k+1}^{\overline{c}=0}(H) &\independent \overline{A}_{D,k+1}^{\overline{c}=0}(H)\mid \overline{A}_{Y,k+1}^{\overline{c}=0}(H), D_{k+1}^{\overline{c}=0} (H)=Y_k^{\overline{c}=0}(H)=0, \overline{L}_k^{\overline{c}=0}(H), \\
    D_{k+1}^{\overline{c}=0}(H) &\independent \overline{A}_{Y,k+1}^{\overline{c}=0}(H) \mid\overline{A}_{D,k+1}^{\overline{c}=0}(H),\nonumber D_{k}^{\overline{c}=0} (H)=Y_k^{\overline{c}=0}(H)=0, \overline{L}_k^{\overline{c}=0}(H),\nonumber\\
    L_{Y,k}^{\overline{c}=0}(H) &\independent \overline{A}_{D,k}^{\overline{c}=0}(H) \mid \overline{A}_{Y,k}^{\overline{c}=0}(H),\nonumber D_{k}^{\overline{c}=0} (H)=Y_k^{\overline{c}=0}(H)=0, L_{D,k}^{\overline{c}=0}(H), \overline{L}_{k-1}^{\overline{c}=0}(H), \nonumber\\
    L_{D,k}^{\overline{c}=0}(H) &\independent \overline{A}_{Y,k}^{\overline{c}=0}(H) \mid\overline{A}_{D,k}^{\overline{c}=0}(H),\nonumber D_{k}^{\overline{c}=0} (H)=Y_k^{\overline{c}=0}(H)=0,  \overline{L}_{k-1}^{\overline{c}=0}(H).  \nonumber
\end{align} 
\end{assumption}

Assumptions \ref{ass:EquivalenceIdentifiability} and \ref{ass:EquivalenceDCC} will allow us to identify the target quantity under the treatment-centered adherence encoding, as shown in the next theorem.

\begin{restatable}{theorem}{thmEquivalence}\label{thm:Equivalence}
Consider a FFRCISTG model with the treatment-centered  encoding. Assume Assumptions \ref{ass:EquivalenceIdentifiability} and \ref{ass:EquivalenceDCC} hold. Then, the counterfactual probability of observing the event of interest is identified by
\begin{align}\label{eq:Equivalence}
        &\mathbb{P}(Y_{K+1}^{\overline{A}_{Y,K+1}=a_Y,\overline{A}_{D,K+1}=a_D,\overline{C}_{K+1}=0}=1)\\
    \nonumber=&\sum_{j=0}^{K}\sum_{\overline{l}_K}  \mathbb{P}  ( Y_{j+1}=1 \mid D_{j+1}=Y_j=C_{j+1}=0,  \overline{L}_{j}=\overline{l}_{j}, \overline{A}_{j+1}=a_Y)  \\
        &  \prod_{s=0}^{j} \bigg[ \mathbb{P} ( D_{s+1}=0 \mid Y_s=D_{s}=C_{s+1}=0, \overline{L}_{s} = \overline{l}_{s}, \overline{A}_{s+1}=a_D) \nonumber \\
        & \times \mathbb{P} ( L_{Y,s}=l_{Y,s} \mid Y_s=D_{s}=C_{s}=0, L_{D,s} =  l_{D,s}, \overline{L}_{s-1} = \overline{l}_{s-1}, \overline{A}_{s}=a_Y)  \nonumber\\
        & \times \mathbb{P} ( L_{D,s}=l_{D,s} \mid Y_s=D_{s}=C_{s}=0, \overline{L}_{s-1} = \overline{l}_{s-1},\overline{A}_{s}=a_D)\nonumber\\
        & \times \mathbb{P}  ( Y_{s}=0 \mid D_{s}=Y_{s-1}=C_{s}=0,  \overline{L}_{s-1}=\overline{l}_{s-1}, \overline{A}_{s}=a_Y) \bigg]\nonumber .
\end{align}
\end{restatable}

\begin{proof}
   This proof follows identical steps to the proof of Theorem \ref{thm:IdentifcationFormula} given in Appendix section  \ref{sec:AppProofIdentification}. Firstly, the modified dismissible component conditions given in Equation \eqref{eq:DCCsEquivalence} imply that for $a_Y,a_D\in\{0,1\}$:
    \begin{align*}
        \mathbb{P}&(Y_{s+1}^{a_Y,a_D,\overline{c}=0}=1\mid D_{s+1}^{a_Y,a_D,\overline{c}=0}=Y_{s}^{a_Y,a_D,\overline{c}=0}=0,\overline{L}_s^{a_Y,a_D,\overline{c}=0}=\overline{l}_s)\\
        =&\mathbb{P}(Y_{s+1}^{a_Y,\overline{c}=0}=1\mid D_{s+1}^{a_Y,\overline{c}=0}=Y_{s}^{a_Y,\overline{c}=0}=0,\overline{L}_s^{a_Y,\overline{c}=0}=\overline{l}_s),\\
        \mathbb{P}&(D_{s+1}^{a_Y,a_D,\overline{c}=0}=1\mid D_{s}^{a_Y,a_D,\overline{c}=0}=Y_{s}^{a_Y,a_D,\overline{c}=0}=0,\overline{L}_s^{a_Y,a_D,\overline{c}=0}=\overline{l}_s)\\
        =&\mathbb{P}(D_{s+1}^{a_D,\overline{c}=0}=1\mid D_{s}^{a_D,\overline{c}=0}=Y_{s}^{a_D,\overline{c}=0}=0,\overline{L}_s^{a_D,\overline{c}=0}=\overline{l}_s),\\
        \mathbb{P}&(L_{Y,s}^{a_Y,a_D,\overline{c}=0}=l_{Y,s}\mid D_{s}^{a_Y,a_D,\overline{c}=0}=Y_{s}^{a_Y,a_D,\overline{c}=0}=0,{L}_{D,s}^{a_Y,a_D,\overline{c}=0}={l}_{D,s},\overline{L}_{s-1}^{a_Y,a_D,\overline{c}=0}=\overline{l}_{s-1})\\
        =&\mathbb{P}(L_{Y,s}^{a_Y,\overline{c}=0}=l_{Y,s}\mid D_{s}^{a_Y,\overline{c}=0}=Y_{s}^{a_Y,\overline{c}=0}=0,{L}_{D,s}^{a_Y,\overline{c}=0}={l}_{D,s},\overline{L}_{s-1}^{a_Y,\overline{c}=0}=\overline{l}_{s-1}),\\
        \mathbb{P}&(L_{D,s}^{a_Y,a_D,\overline{c}=0}=l_{D,s}\mid D_{s}^{a_Y,a_D,\overline{c}=0}=Y_{s}^{a_Y,a_D,\overline{c}=0}=0,\overline{L}_{s-1}^{a_Y,a_D,\overline{c}=0}=\overline{l}_{s-1})\\
        =&\mathbb{P}(L_{D,s}^{a_D,\overline{c}=0}=l_{D,s}\mid D_{s}^{a_D,\overline{c}=0}=Y_{s}^{a_D,\overline{c}=0}=0,\overline{L}_{s-1}^{a_D,\overline{c}=0}=\overline{l}_{s-1}), 
    \end{align*}
    where the intervention notation $a_Y,a_D,\overline{c}=0$ denotes $\overline{A}_Y=a_Y,\overline{A}_D=a_D,\overline{c}=0$, while $a,\overline{c}=0$ stands for $\overline{A}=a,\overline{c}=0$.   The proof of these equalities follows exactly the steps of the proof of \cite[Lemma 1, App. B]{stensrud2021generalized}, with the only difference of taking into account the time structure over $A_Y$ and $A_D$.    Analogously,  conditions \eqref{eq:ExchangeabilityEquivalence}-\eqref{eq:ConsistencyEquivalence} imply that for $s=0,\ldots,K$ and $z\in\{0,1\}$ we have that:\\
    
    \resizebox{\linewidth}{!}{$\begin{aligned}
        \mathbb{P}(&Y_{s+1}^{\overline{A}=a,\overline{c}=0}=1\mid D_{s+1}^{\overline{A}=a,\overline{c}=0}=Y_{s}^{\overline{A}=a,\overline{c}=0}=0,\overline{L}_s^{\overline{A}=a,\overline{c}=0}=\overline{l}_s)\\&=\mathbb{P}(Y_{s+1}=1\mid C_{s+1}=D_{s+1}=Y_s=0,\overline{L}_s=\overline{l}_s,\overline{A}_{s+1}=a),\\
        \mathbb{P}(&D_{s+1}^{\overline{A}=a,\overline{c}=0}=1\mid D_{s}^{\overline{A}=a,\overline{c}=0}=Y_{s}^{\overline{A}=a,\overline{c}=0}=0,\overline{L}_s^{\overline{A}=a,\overline{c}=0}=\overline{l}_s)\\&=\mathbb{P}(D_{s+1}=1\mid C_{s+1}=D_{s}=Y_s=0,\overline{L}_s=\overline{l}_s,\overline{A}_{s+1}=a),\\
        \mathbb{P}(&L_{Y,s}^{\overline{A}=a,\overline{c}=0}=l_{Y,s}\mid D_{s}^{\overline{A}=a,\overline{c}=0}=Y_{s}^{\overline{A}=a,\overline{c}=0}=0,\overline{L}_{s-1}^{\overline{A}=a,\overline{c}=0}=\overline{l}_{s-1},{L}_{D,s}^{\overline{A}=a,\overline{c}=0}={l}_{D,s})\\&=\mathbb{P}(L_{Y,s}=l_{Y,s}\mid C_{s}=D_{s}=Y_s=0,\overline{L}_{s-1}=\overline{l}_{s-1},{L}_{D,s}={l}_{D,s},\overline{A}_{s}=a),\\
        \mathbb{P}(&L_{D,s}^{\overline{A}=a,\overline{c}=0}=l_{D,s}\mid D_{s}^{\overline{A}=a,\overline{c}=0}=Y_{s}^{\overline{A}=a,\overline{c}=0}=0,\overline{L}_{s-1}^{\overline{A}=a,\overline{c}=0}=\overline{l}_{s-1})\\&=\mathbb{P}(L_{D,s}=l_{D,s}\mid C_{s}=D_{s}=Y_s=0,\overline{L}_{s-1}=\overline{l}_{s-1},\overline{A}_{s}=a).
    \end{aligned}
    $}\vspace{1em}
    The proof of these equalities is analogous to the proof of Lemma \ref{lemma:Identification2} in Appendix \ref{sec:AppProofIdentification}. At this point, the identification formula under the treatment-centered adherence encoding \eqref{eq:Equivalence} follows as: \\
    \resizebox{\linewidth}{!}{$ 
    \begin{aligned}\allowdisplaybreaks
        &\mathbb{P}(Y_{K+1}^{a_Y,a_D,\overline{c}=0}=1)\\=&\sum_{\overline{l}_K}\mathbb{P}(Y_{K+1}^{a_Y,a_D,\overline{c}=0}=1,\overline{L}_K^{a_Y,a_D,\overline{c}=0}=\overline{l}_K)\\
        =&\sum_{j=0}^{K}\sum_{\overline{l}_K}  \mathbb{P}  ( Y_{j+1}^{a_Y,a_D,\overline{c}=0}=1 \mid D_{j+1}^{a_Y,a_D,\overline{c}=0}=Y_j^{a_Y,a_D,\overline{c}=0}=0,  \overline{L}^{a_Y,a_D,\overline{c}=0}_{j}=\overline{l}_{j})  \\
        &  \prod_{s=0}^{j} \bigg[ \mathbb{P} ( D_{s+1}^{a_Y,a_D,\overline{c}=0}=0 \mid Y_s^{a_Y,a_D,\overline{c}=0}=D_{s}^{a_Y,a_D,\overline{c}=0}=0, \overline{L}_{s}^{a_Y,a_D,\overline{c}=0} = \overline{l}_{s})  \\
        & \times \mathbb{P} ( L_{Y,s}^{a_Y,a_D,\overline{c}=0}=l_{Y,s} \mid Y_s^{a_Y,a_D,\overline{c}=0}=D_{s}^{a_Y,a_D,\overline{c}=0}=0, L_{D,s}^{a_Y,a_D,\overline{c}=0} =  l_{D,s}, \overline{L}_{s-1}^{a_Y,a_D,\overline{c}=0} = \overline{l}_{s-1})  \\
        & \times \mathbb{P} ( L_{D,s}^{a_Y,a_D,\overline{c}=0}=l_{D,s} \mid Y_s^{a_Y,a_D,\overline{c}=0}=D_{s}^{a_Y,a_D,\overline{c}=0}=0, \overline{L}_{s-1}^{a_Y,a_D,\overline{c}=0} = \overline{l}_{s-1})\\
        & \times \mathbb{P}  ( Y_{s}^{a_Y,a_D,\overline{c}=0}=0 \mid D_{s}^{a_Y,a_D,\overline{c}=0}=Y_{s-1}^{a_Y,a_D,\overline{c}=0}=0,  \overline{L}_{s-1}^{a_Y,a_D,\overline{c}=0}=\overline{l}_{s-1}) \bigg]\\
        =&\sum_{j=0}^{K}\sum_{\overline{l}_K}  \mathbb{P}  ( Y_{j+1}^{a_Y,\overline{c}=0}=1 \mid D_{j+1}^{a_Y,\overline{c}=0}=Y_j^{a_Y,\overline{c}=0}=0,  \overline{L}^{a_Y,\overline{c}=0}_{j}=\overline{l}_{j})  \\
        &  \prod_{s=0}^{j} \bigg[ \mathbb{P} ( D_{s+1}^{a_D,\overline{c}=0}=0 \mid Y_s^{a_D,\overline{c}=0}=D_{s}^{a_D,\overline{c}=0}=0, \overline{L}_{s}^{a_D,\overline{c}=0} = \overline{l}_{s})  \\
        & \times \mathbb{P} ( L_{Y,s}^{a_Y,\overline{c}=0}=l_{Y,s} \mid Y_s^{a_Y,\overline{c}=0}=D_{s}^{a_Y,\overline{c}=0}=0, L_{D,s}^{a_Y,\overline{c}=0} =  l_{D,s}, \overline{L}_{s-1}^{a_Y,\overline{c}=0} = \overline{l}_{s-1})  \\
        & \times \mathbb{P} ( L_{D,s}^{a_D,\overline{c}=0}=l_{D,s} \mid Y_s^{a_D,\overline{c}=0}=D_{s}^{a_D,\overline{c}=0}=0, \overline{L}_{s-1}^{a_D,\overline{c}=0} = \overline{l}_{s-1})\\
        & \times \mathbb{P}  ( Y_{s}^{a_Y,\overline{c}=0}=0 \mid D_{s}^{a_Y,\overline{c}=0}=Y_{s-1}^{a_Y,\overline{c}=0}=0,  \overline{L}_{s-1}^{a_Y,\overline{c}=0}=\overline{l}_{s-1}) \bigg]\\
        =&\sum_{j=0}^{K}\sum_{\overline{l}_K}  \mathbb{P}  ( Y_{j+1}=1 \mid D_{j+1}=Y_j=C_{j+1}=0,  \overline{L}_{j}=\overline{l}_{j},\overline{A}_{j+1}=a_Y)  \\
        &  \prod_{s=0}^{j} \bigg[ \mathbb{P} ( D_{s+1}=0 \mid Y_s=D_{s}=C_{s+1}=0, \overline{L}_{s} = \overline{l}_{s},\overline{A}_{s+1}=a_D)  \\
        & \times \mathbb{P} ( L_{Y,s}=l_{Y,s} \mid Y_s=D_{s}=C_{s}=0, L_{D,s} =  l_{D,s}, \overline{L}_{s-1} = \overline{l}_{s-1},  \overline{A}_{s}=a_Y)  \\
        & \times \mathbb{P} ( L_{D,s}=l_{D,s} \mid Y_s=D_{s}=C_{s}=0, \overline{L}_{s-1} = \overline{l}_{s-1},\overline{A}_{s}=a_D)\\
        & \times \mathbb{P}  ( Y_{s}=0 \mid D_{s}=Y_{s-1}=C_{s}=0,  \overline{L}_{s-1}=\overline{l}_{s-1},  \overline{A}_{s}=a_Y) \bigg].
    \end{aligned}
    $}\vspace{1em}
    This concludes the proof.
\end{proof}
The similarities between the identification formulas \eqref{eq:IdentificationFormula} and \eqref{eq:Equivalence} will allow us to clarify in which sense the strategy-centered and treatment-centered adherence encodings are equivalent. This becomes explicit in the following corollary.
\begin{restatable}{corollary}{coroEquivalence}\label{coro:Equivalence}
    Under the conditions of Theorems \ref{thm:IdentifcationFormula} and \ref{thm:Equivalence}, we have that $$\mathbb{P}(Y_{K+1}^{Z_Y=z_Y,Z_D=z_D,\overline{C}_{K+1}=0,\overline{R}_{K+1}=1}=1)=\mathbb{P}(Y_{K+1}^{\overline{A}_{Y,K+1}=z_Y,\overline{A}_{D,K+1}=z_D,\overline{C}_{K+1}=0}=1),$$
    meaning both adherence encodings equivalently describe the probability of observing the event of interest.
\end{restatable}
\begin{proof}
    The equality follows from the identification formulas \eqref{eq:IdentificationFormula} and \eqref{eq:Equivalence}, and the equivalence between the two adherence encodings for interventions on perfect adherence to a certain treatment strategy in the two-arm trial: $\{\overline{A}_k=z\}\Leftrightarrow \{Z=z,\ \overline{R}_k=1\}$ for all $k$ in $\{1,\ldots,K+1\}$ and $z\in\{0,1\}$.
\end{proof}

This corollary  clarifies the sense of equivalence we previously mentioned: the  probabilities of observing the event of interest under a certain sustained treatment strategy in the four-arm trial are identified by  quantities in the two-arm trial which differ only on a reparametrization of the adherence encoding.

\subsection{On the equivalence of the identification conditions under both adherence encodings}\label{sec:AppEquivalenceIdentificationConditions}
Corollary \ref{coro:Equivalence} requires the identifiability conditions under both adherence encodings to hold. It could be argued that if both adherence encodings lead to equivalent identification results for the probabilities of interest, the identifiability conditions should also be equivalent. In this section we elaborate on this argument. Consider first the identifiability conditions under the strategy-centered encoding stated in the two-arm trial: Equations \eqref{eq:Exchangeability}-\eqref{eq:Positivity}. These are translated onto the treatment-centered encoding as conditions \eqref{eq:ExchangeabilityEquivalence}-\eqref{eq:ConsistencyEquivalence}. In the next proposition we show that these sets of conditions are equivalent.

\begin{proposition}
    Under a FFRCISTG model conditions \eqref{eq:Exchangeability}-\eqref{eq:Positivity} and conditions \eqref{eq:ExchangeabilityEquivalence}- \eqref{eq:ConsistencyEquivalence} are equivalent.
\end{proposition}
\begin{proof}
    As per the equivalence of the perfect adherence intervention in  both treatment encodings $\{\overline{A}_{K+1}=z\}\Leftrightarrow \{Z=z,\ \overline{R}_{K+1}=1\}$, the consistency conditions \eqref{eq:Consistency} and \eqref{eq:ConsistencyEquivalence} are equivalent. Keeping this equivalence in mind, if we evaluate \eqref{eq:ExchangeabilityEquivalence} for $k=0$  we obtain a condition which implies \eqref{eq:ExchangeabilityPast} and \eqref{eq:ExchangeabilityFuture} when $k=0$, as in $R_1=1$ almost surely. Conversely, condition  \eqref{eq:ExchangeabilityFuture} implies \eqref{eq:ExchangeabilityEquivalence} when $k\geq 1$, and when $k=0$ the same implication holds taking \eqref{eq:ExchangeabilityPast} into account.

    Regarding positivity, \eqref{eq:PositivityEquivalence} evaluated when $k=0$ implies \eqref{eq:PositivityL0} after marginalizing over $C_1$. Conditions   \eqref{eq:PositivityEquivalence} and \eqref{eq:PositivityCR} are equivalent for all $k$ in $\{0,\ldots,K\}$ as after conditioning on  $Z=z $ (equivalently $A_1=z$) we have that $R_{k}=1\Leftrightarrow A_k=z$. Lastly, conditions \eqref{eq:PositivityEquivalenceDCC} and \eqref{eq:PositivityA} are equivalent as $\overline{R}_k=1$ is equivalent to $A_1=\ldots=A_k$.
\end{proof}

Lastly we consider the dismissible component conditions. The next result shows that Algorithm \ref{alg:DagDrawing} provides the correct way of constructing the extended causal DAG under the strategy-centered encoding, as the dismissible component conditions under both treatment encodings become equivalent. 

\begin{proposition}\label{prop:EquivalenceDCCd-sep}
    Let $\mathcal{G}$ be the extended causal DAG under the treatment-centered encoding. Let $\mathcal{G}^*$ the extended causal DAG under the strategy-centered encoding be constructed from $\mathcal{G}$ applying Algorithm \ref{alg:DagDrawing}. Then Assumptions \ref{ass:DCC} and  \ref{ass:EquivalenceDCC}  are equivalent when they are read as $d$-separation statements in the associated SWIGs.
\end{proposition}
\begin{proof}
    We present here the proof of the equivalence of the first $d$-separation statements, as the rest follow analogously.  Recall that $d$-separation in a SWIG is read as standard $d$-separation in the DAG obtained by removing all fixed nodes from the SWIG; i.e. by considering only paths containing random nodes (See \cite[Sec. 3.5.2]{richardson2013single} for details). We provide now the equivalence proof, both implications by a counter-reciprocate argument. For the sake of notation, denote by  $\mathcal{G}_S$ and $\mathcal{G}_S^*$ the SWIGs related to the DCCs under the treatment- and strategy-centered encodings respectively.\\
    ``$\Leftarrow$'' Assume there is a path in $\mathcal{G}_S^*$ between $Z_D$ and $Y_{k+1}^{\cor}$ open given $B:=\{Z_Y,\overline{D}_{k+1}^{\cor},\overline{Y}_{k}^{\cor}, \overline{L}_{k}^{\cor}\}$. Firstly, note that this path cannot go further in time than $k+1$, as otherwise to reach $Y_{k+1}^{\cor}$ there would have to be a collider without descendants in $B$ which would block the path. Second, it can not intersect $\overline{R}_{K+1}^{\cor}$ or $\overline{C}_{K+1}^{\cor}$, as these random nodes are colliders without any descendants which would block the path. Therefore this open path only contains intermediate nodes in $B$. As it is open, if there are any intermediate nodes they must all be colliders. Thus, either there are no intermediate nodes or there is a single one which is a collider in $B$. But this second possibility can also not happen because the edge exiting this intermediate node would have to point forward in time, exiting the intermediate node, and making it not a collider. Therefore, this path ought to be $Z_D\to Y_{k+1}^{\cor}$. As this edge exists in $\mathcal{G}_S^*$ between two random nodes which are not split, it also exists in $\mathcal{G}^*$. But due to Algorithm \ref{alg:DagDrawing}, said edge exists in $\mathcal{G}^*$ if and only if $Y_{k+1}$ is a child of a node in $\overline{A}_{D,k+1}$ in $\mathcal{G}$, what would then violate $d$-separation in $\mathcal{G}_S$. This is a contradiction.\\
    ``$\Rightarrow$'' Assume there is a path between $Y_{k+1}^{\overline{c}=0}$ and $\overline{A}_{D,k+1}^{\overline{c}=0}$ in $\mathcal{G}_S$ open given $Q:=\{ \overline{A}_{Y,k+1}^{\overline{c}=0}, \overline{D}_{k+1}^{\overline{c}=0}, \overline{Y}_k^{\overline{c}=0}, \overline{L}_k^{\overline{c}=0}\}$. Assume that this path is between $A_{D,t}^{\overline{c}=0}$ and $Y_{k+1}^{\overline{c}=0}$ for some $t$ in $\{1,\ldots,k+1\}$. W.l.o.g. we assume that this path does not intersect $\overline{A}_{D,k+1}^{\overline{c}=0}\backslash A_{D,t}^{\overline{c}=0}$, as otherwise we could take the last node of this set the path intersects when traversed towards $Y_{k+1}^{\overline{c}=0}$. Analogous reasoning as before leads to the conclusion that $A_{D,t}^{\overline{c}=0}\to Y_{k+1}^{\overline{c}=0}$ exists in $\mathcal{G}_S$. As these random nodes are not split, $A_{D,t}\to Y_{k+1}$ exists in $\mathcal{G}$. By Algorithm \ref{alg:DagDrawing}, $Z_D\to Y_{k+1}$ exists in $\mathcal{G}^*$, and consequently $Z_D\to Y_{k+1}^{\cor}$ exists in $\mathcal{G}_S^*$, which is a contradiction. This concludes the proof.
\end{proof}

This proposition does not establish equivalence between the dismissible component conditions, as these are conditional independencies, and we have proven equivalence of $d$-separation statements. The missing component will be a faithfulness assumption, as we can see in the next Corollary.

\begin{corollary}\label{coro:EquivalenceDCC}
    Let $\mathcal{G}$ be the extended causal DAG under the treatment-centered encoding. Let the extended causal DAG under the strategy-centered encoding be constructed from $\mathcal{G}$ applying Algorithm \ref{alg:DagDrawing}. Consider  the SWIGs associated to these DAGs under the interventions considered in Assumptions \ref{ass:DCC} and  \ref{ass:EquivalenceDCC} respectively, and assume that under the respective FFRCISTG models faithfulness holds. Then the dismissible component conditions under both treatment encodings (Assumptions \ref{ass:DCC} and  \ref{ass:EquivalenceDCC}) are equivalent.
\end{corollary}
\begin{proof}
    It follows from the SWIG global Markov property under a FFRCISTG model \cite[Sec. 3.5.2]{richardson2013single}, the faithfulness assumption \cite[Sec. 3.6.5]{richardson2013single}, and Proposition \ref{prop:EquivalenceDCCd-sep}.
\end{proof}

To sum up, we see how the treatment-centered encoding gives an equivalent identification result: under certain assumptions, the identification formulas under both treatment encodings differ only in the adherence reparametrization.

\section{Estimation under the strategy-centred encoding}\label{sec:AppEstimation}

\subsection{An alternative weighted estimator}\label{sec:AppEstimationWeighted}
\begin{restatable}{theorem}{thmWeigthedFormulaApp}\label{thm:WeigthedFormulaApp}
 Under the conditions of Theorem \ref{thm:IdentifcationFormula} an  equivalent  identification formula is:
    \begin{align}\label{eq:WeigthedFormulaApp}
        &\mathbb{P}(Y_{K+1}^{z_Y,z_D,\cor}=1)\\
         =&\sum_{s=0}^K\mathbb{E}  [ W_{(C,R),s}(z_Y) W_{D,s}  W_{L_{D},s}  (1-Y_s)(1-D_{s+1}) Y_{s+1} \mid Z=z_Y], \nonumber
    \end{align}
where 
\begin{align*}
W_{D,s}&=\prod_{j=0}^{s} \frac{  \mathbb{P}(D_{j+1}=0 \mid C_{j+1}=D_{j}=Y_j= 0, \overline{L}_{j},  Z = z_D,\overline{R}_j=1) }{  \mathbb{P}(D_{j+1}=0 \mid  C_{j+1}=D_{j}=Y_j= 0, \overline{L}_{j},  Z = z_Y,\overline{R}_j=1) }, \\
 W_{L_{D},s}  &=\prod_{j=0}^{s}  \frac{  \mathbb{P}(L_{D,j} = l_{D,j} \mid C_{j}=  D_{j} =Y_j= 0, \overline{L}_{j-1},  Z = z_D, \overline{R}_{j-1}=1) }{  \mathbb{P}(L_{D,j} = l_{D,j} \mid C_{j}=  D_{j} =Y_j= 0, \overline{L}_{j-1},  Z = z_Y, \overline{R}_{j-1}=1) }, 
\end{align*}
and $W_{(C,R),s} (z)$ as defined in Theorem \ref{thm:WeigthedFormulas}.

\end{restatable}
\begin{proof}
    Analogous to the proof of Theorem \ref{thm:WeigthedFormulas} presented in Appendix \ref{sec:AppProofEstimation}.
\end{proof}
 The same comments made on Section \ref{sec:WeightedEstimators} about how this identification expression relates sustained effects to artificial censoring apply here. Now, analogously to what was done in Section \ref{sec:WeightedEstimators}, based now on Equation \eqref{eq:WeigthedFormulaApp} we can define another weighted estimator of $\nu$ for which we specify the  conditional distributions  which appear in $W_D, W_{L_D}, W_{(C,R)}$ , yielding:
\begin{equation*}
    \widehat{\nu}_{weighted,D}=\widehat{\mathbb{E}}_n\left[  \frac{I(Z=z_Y)}{\widehat{\mathbb{P}}_n(Z=z_Y)}\sum_{s=0}^K\Tilde{W}_{(C,R),s}(z_Y) \Tilde{W}_{D,s}  \Tilde{W}_{L_{D},s}  (1-Y_s)(1-D_{s+1}) Y_{s+1}\right].
\end{equation*}
Identically to the weighted estimator presented in the main text, under the conditions of Theorem \ref{thm:IdentifcationFormula}, $\widehat{\nu}_{weighted,D}$ is a consistent estimator of $\mathbb{P}(Y_{K+1}^{z_Y,z_D,\cor}=1)$ provided that the postulated models are correctly specified.

\subsection{The one-step estimator}\label{sec:AppEstimationDR}
\begin{definition}[Influence function of an estimand, adapted from \cite{tsiatis2007semiparametric}] \label{def:IF}
    We define the influence function $\chi^1$ of a (regular asymptotically linear) estimand $\chi$ as a random variable with mean zero and finite variance such that for every (regular) parametric submodel $\{\mathbb{P}_t\colon t\in[0,1)\}$ it satisfies $$\frac{d\chi(\mathbb{P}_t)}{dt}\bigg |_{t=0}=\mathbb{E}[\chi^1g],$$ where $g$ is the score of the true law $\mathbb{P}_0$.
\end{definition}

Further details on semi-parametric methods, the computation of influence functions, and the asymptotic properties of one-step estimators can be found in \cite{hines2022demystifying}, or with greater detail in \cite{tsiatis2007semiparametric}. As in Section \ref{sec:Estimation}, denote by $\nu$ the right-hand-side of the identification formula \eqref{eq:IdentificationFormula}. In order to compute the one-step estimator, the one missing quantity  is $\nu$'s influence function $\nu^1$, which we can see in the next theorem.

\begin{restatable}{theorem}{thmInfluenceFucntion}\label{thm:InfluenceFucntion}
    The influence function of  $\nu$ is:
    \begin{align}\label{eq:InfluenceFunction}
        &\nu^1=\sum_{s=0}^K \nu_s^1=\\
            &\sum_{s=0}^K\left[\Omega_s(1-D_s)(1-Y_s)\left\{\frac{RR(D_{s+1})}{\pi_{s+1}(z_Y)}(1-D_{s+1})I(\phi_{s+1}(z_Y))(Y_{s+1}-h^{*Y}_{s+1}(1))\right.\right.\nonumber\\
            &\left.+I(\phi_{s}(z_Y))(h^D_{s+1}(h^{*Y}_{s+1}(1))-h^{L_Y}_{s}(h^D_{s+1}(h^{*Y}_{s+1}(1)))) \right\}\nonumber\\
            &+\Lambda_s(1-D_s)(1-Y_s)\left\{\frac{RR(L_{Y,s})}{\pi_{s+1}(z_D)}I(\phi_{s+1}(z_D))\bigg(h^{*Y}_{s+1}(1)(1-D_{s+1})\right.\nonumber\\&-h^D_{s+1}(h^{*Y}_{s+1}(1))\bigg)\left.\nonumber
            +I(\phi_{s}(z_D))(h^{L_Y}_{s}(h^D_{s+1}(h^{*Y}_{s+1}(1)))-T^{(s+1)}_{s+1})\right\}\nonumber\\
            &+\sum_{j=1}^s\left[ \Omega_{j-1}(1-D_{j-1})(1-Y_{j-1})\left\{\frac{RR(D_{j})}{\pi_{j}(z_Y)}(1-D_{j})I(\phi_{j}(z_Y))\bigg(T^{(s+1)}_{j+1}(1-Y_{j})\nonumber\right.\right.\\&-h^{Y}_{j}(T^{(s+1)}_{j+1})\bigg)
            \left.+I(\phi_{j-1}(z_Y))(h^D_{j}(h^{Y}_{j}(T^{(s+1)}_{j+1}))-h^{L_Y}_{j-1}(h^D_{j}(h^{Y}_{j}(T^{(s+1)}_{j+1}))))\right\}\nonumber\\
            &+\Lambda_{j-1}(1-D_{j-1})(1-Y_{j-1})\left\{\frac{RR(L_{Y,j-1})}{\pi_{j}(z_D)}I(\phi_{j}(z_D))\bigg(h^{Y}_{j}(T^{(s+1)}_{j+1})(1-D_{j})\right.\nonumber\\
            &-h^D_{j}(h^{Y}_{j}(T^{(s+1)}_{j+1}))\bigg)\left.\left.\left.+I(\phi_{j-1}(z_D))(h^{L_Y}_{j-1}(h^D_{j}(h^{Y}_{j}(T^{(s+1)}_{j+1})))-T^{(s+1)}_{j})\right\}\right]\right],\nonumber
    \end{align}
    where we understand $\sum_{j=1}^0$ to be empty, and we have defined: 
    \begin{align*}
        \phi_j(z)&=\{Z=z,C_j=0,\overline{R}_{j}=1\},\\
        RR(Y_j)&=\frac{\mathbb{P}(Y_{j}=0 \mid D_{j}=Y_{j-1}= 0, \overline{L}_{j-1}, \phi_j(z_Y)) }{   \mathbb{P}(Y_{j}=0 \mid D_{j}=Y_{j-1}= 0, \overline{L}_{j-1},  \phi_j(z_D)) },\\
        RR(L_{Y,j})&=\frac{\mathbb{P}(L_{Y,j}=l_{Y,j} \mid D_{j}=Y_{j}= 0, L_{D,j}=l_{D,j}, \overline{L}_{j-1}, \phi_{j}(z_Y)) }{\mathbb{P}(L_{Y,j}=l_{Y,j} \mid D_{j}=Y_{j}= 0, L_{D,j}=l_{D,j}, \overline{L}_{j-1}, \phi_{j}(z_D)) },\\
        RR(D_j)&=\frac{\mathbb{P}(D_{j}=0 \mid D_{j-1}=Y_{j-1}= 0, \overline{L}_{j-1}, \phi_j(z_D)) }{   \mathbb{P}(D_{j}=0 \mid D_{j-1}=Y_{j-1}= 0, \overline{L}_{j-1},  \phi_j(z_Y)) },\\
        RR(L_{D,j})&=\frac{\mathbb{P}(L_{D,j}=l_{D,j} \mid D_{j}=Y_{j}= 0, \overline{L}_{j-1}, \phi_{j}(z_D)) }{   \mathbb{P}(L_{D,j}=l_{D,j} \mid D_{j}=Y_{j}= 0, \overline{L}_{j-1}, \phi_{j}(z_Y)) },\\
        \pi_j(z)&=\mathbb{P}(C_j=0,R_{j}=1\mid D_{j-1}=Y_{j-1}=0,\overline{L}_{j-1},\phi_{j-1}(z)),\\
        \Omega_j&=\frac{1}{\mathbb{P}(Z=z_Y)}\prod_{s=0}^j\frac{RR(D_s)RR(L_{D,s})}{\pi_s(z_Y)},\\
        \Lambda_j&=\frac{1}{\mathbb{P}(Z=z_D)}\prod_{s=0}^j\frac{RR(L_{Y,s-1})RR(Y_s)}{\pi_s(z_D)},\\
        h^Y_j(x)&=\mathbb{E}[x(1-Y_j)\mid D_{j}, Y_{j-1}, \overline{L}_{j-1}, \phi_j(z_Y) ], \\
        h^{*Y}_j(x)&=\mathbb{E}[xY_j\mid D_{j}, Y_{j-1}, \overline{L}_{j-1}, \phi_j(z_Y) ] ,\\
        h^D_j(x)&=\mathbb{E}[x(1-D_j)\mid D_{j-1}, Y_{j-1}, \overline{L}_{j-1}, \phi_j(z_D) ] ,\\
        h^{L_Y}_j(x)&=\mathbb{E}[x\mid D_{j}, Y_{j},L_{D,j}, \overline{L}_{j-1}, \phi_{j-1}(z_Y) ], \\
        h^{L_D}_j(x)&=\mathbb{E}[x\mid D_{j}, Y_{j}, \overline{L}_{j-1}, \phi_{j-1}(z_D) ], \\
        T^{(s)}_j&=\begin{cases}
            h^{L_D}_{s-1}(h^{L_Y}_{s-1}(h^D_{s}(h^{*Y}_{s}(1)))),\quad j=s\\
            h^{L_D}_{j-1}(h^{L_Y}_{j-1}(h^D_{j}(h^{Y}_{j}(T^{(s)}_{j+1})))), \quad\text{recursively for }j=s-1,\ldots,1
        \end{cases}.
    \end{align*}
\end{restatable}
\begin{proof}
    The identification formula \eqref{eq:IdentificationFormula} allows to express the target quantity $\nu$ as a sum of individual  terms $\nu=\sum_{s=0}^K\nu_s$, which corresponds to the probability of observing the target event at each time point. This additivity translates to the influence function scale, as $$\frac{d\nu(\mathbb{P}_t)}{dt}\bigg|_{t=0}=\sum_{s=0}^K \frac{d\nu_s(\mathbb{P}_t)}{dt}\bigg|_{t=0},$$ and therefore $\nu^1=\sum_{s=0}^K\nu_s^1$. Thus, for the remainder of this proof we focus on $\nu_s$, which  can be read off equation \eqref{eq:IdentificationFormula}. The $\nu_{s}$ term  can be expressed as an iterated conditional expectation (ICE): 
    \begin{align*}
        &\mathbb{E}[\ldots[\mathbb{E}[\mathbb{E}[\mathbb{E}[\mathbb{E}[\mathbb{E}[\mathbb{E}[Y_{s+1}| D_{s+1},Y_s,\overline{L}_s,\phi_{s+1}(z_Y)](1-D_{s+1})|D_{s},Y_s,\overline{L}_s,\phi_{s+1}(z_D)]\\&|D_{s},Y_s,L_{D,s}\overline{L}_{s-1},\phi_{s}(z_Y)]|D_{s},Y_s,\overline{L}_{s-1},\phi_{s}(z_D)](1-Y_s)|D_{s},Y_{s-1},\overline{L}_{s-1},\phi_{s}(z_Y)]\\&\times(1-D_{s})|D_{s-1},Y_{s-1},\overline{L}_{s-1},\phi_{s}(z_D)]\ldots]=\\
        &=h_0^{L_D}(h_0^{L_Y}(h_1^D(h_1^Y(h_1^{L_D}(h_1^{L_Y}(h_2^D(\ldots(h_{s+1}^D(h_{s+1}^{*Y}(1)))\ldots)))))))=T^{(s+1)}_1.
    \end{align*}
    Expressing $\nu_s$ as an ICE will  simplify the computation of the influence function by exploiting the recursive definition of the $T_j^{(s)}$'s. Indeed, for $j<s$ we have:
    \begin{align*}
        T_j^{(s)}&=h^{L_D}_{j-1}(h^{L_Y}_{j-1}(h^D_{j}(h^{Y}_{j}(T^{(s)}_{j+1}))))\\
        &=\mathbb{E}[\mathbb{E}[\mathbb{E}[\mathbb{E}[T_{j+1}^{(s)}(1-Y_j)|D_{j},Y_{j-1} ,\overline{L}_{j-1},\phi_{j}(z_Y)](1-D_j)|D_{j-1},Y_{j-1},\overline{L}_{j-1},\\&\qquad\phi_{j}(z_D)]|D_{j-1},Y_{j-1},L_{D,j-1},\overline{L}_{j-2},\phi_{j-1}(z_Y)]|D_{j-1},Y_{j-1},\overline{L}_{j-2},\phi_{j-1}(z_D)].
    \end{align*}
    Therefore,  we also find the recursion: 
    \begin{subequations}\label{eq:PortDerivRecursive}
        \begin{align}
        \frac{dT^{(s)}_j(\mathbb{P}_t)}{dt}\bigg|_{t=0}&=\nonumber\\&= h^{L_D}_{j-1}(h^{L_Y}_{j-1}(h^D_{j}(h_j^Y(T^{(s)}_{j+1}g_{Y_j|D_{j},Y_{j-1} ,\overline{L}_{j-1},\phi_{j}(z_Y)}))))\label{eq:PortDerivAux1}\\
        &+h^{L_D}_{j-1}(h^{L_Y}_{j-1}(h^D_{j}(h^{Y}_{j}(T^{(s)}_{j+1})g_{D_j|D_{j-1},Y_{j-1},\overline{L}_{j-1},\phi_{j}(z_D)})))\label{eq:PortDerivAux2}\\
        &+h^{L_D}_{j-1}(h^{L_Y}_{j-1}(h^D_{j}(h^{Y}_{j}(T^{(s)}_{j+1}))g_{L_{Y,j-1}|D_{j-1},Y_{j-1},L_{D,j-1},\overline{L}_{j-2},\phi_{j-1}(z_Y)}))\label{eq:PortDerivAux3}\\
        &+h^{L_D}_{j-1}(h^{L_Y}_{j-1}(h^D_{j}(h^{Y}_{j}(T^{(s)}_{j+1})))g_{L_{D,j-1}|D_{j-1},Y_{j-1},\overline{L}_{j-2},\phi_{j-1}(z_D)})\label{eq:PortDerivAux4}\\
        &+h^{L_D}_{j-1}(h^{L_Y}_{j-1}(h^D_{j}(h^{Y}_{j}\left(\frac{dT^{(s)}_{j+1}(\mathbb{P}_t)}{dt}\bigg|_{t=0}\right)))),\nonumber
    \end{align}
    \end{subequations}
    where $g_{X|U}$ denotes the score function of $X|U$. As we are particularly interested in the $j=1$ case, we will now study each of the addends  \eqref{eq:PortDerivAux1}-\eqref{eq:PortDerivAux4} separately in this case. Assuming $s>1$, for $j=1$:
    \begin{align*}
        \eqref{eq:PortDerivAux1}&= h^{L_D}_{0}(h^{L_Y}_{0}(h^D_{1}(\mathbb{E}[((1-Y_1)T^{(s)}_{2}-h_1^Y(T^{(s)}_{2}))g_{Y_1|D_{1} ,\overline{L}_{0},\phi_{1}(z_Y)}|D_{1} ,\overline{L}_{0},\phi_{1}(z_Y)])))\\
        &=\mathbb{E}[\rho^{L_D}_0\rho^{L_Y}_0\rho^D_1\frac{(1-D_1)I(\phi_1(z_Y))}{\mathbb{P}(\phi_1(z_Y)|D_1=0,L_0)}((1-Y_1)T^{(s)}_{2}-h_1^Y(T^{(s)}_{2}))g_{Y_1|D_{1} ,\overline{L}_{0},\phi_{1}(z_Y)}]\\
        &=\mathbb{E}[\rho^{L_D}_0\rho^{L_Y}_0\rho^D_1\frac{(1-D_1)I(\phi_1(z_Y))}{\mathbb{P}(\phi_1(z_Y)|D_1=0,L_0)}((1-Y_1)T^{(s)}_{2}-h_1^Y(T^{(s)}_{2}))g],\\
        \eqref{eq:PortDerivAux2}&= h^{L_D}_{0}(h^{L_Y}_{0}(\mathbb{E}[(h^{Y}_{1}(T^{(s)}_{2})(1-D_1)-h^D_{1}(h^{Y}_{1}(T^{(s)}_{2})))g_{D_1|\overline{L}_{0},\phi_{1}(z_D)}|\overline{L}_{0},\phi_{1}(z_D)]))\\
        &=\mathbb{E}[\rho^{L_D}_0\rho^{L_Y}_0\frac{I(\phi_1(z_D))}{\mathbb{P}(\phi_1(z_D)|L_0)}(h^{Y}_{1}(T^{(s)}_{2})(1-D_1)-h^D_{1}(h^{Y}_{1}(T^{(s)}_{2})))g_{D_1|\overline{L}_{0},\phi_{1}(z_D)}]\\
        &= \mathbb{E}[\rho^{L_D}_0\rho^{L_Y}_0\frac{I(\phi_1(z_D))}{\mathbb{P}(\phi_1(z_D)|L_0)}(h^{Y}_{1}(T^{(s)}_{2})(1-D_1)-h^D_{1}(h^{Y}_{1}(T^{(s)}_{2})))g],\\
        \eqref{eq:PortDerivAux3}&= h^{L_D}_{0}(\mathbb{E}[(h^D_{1}(h^{Y}_{1}(T^{(s)}_{2}))-h^{L_Y}_{0}(h^D_{1}(h^{Y}_{1}(T^{(s)}_{2}))))g_{L_{Y,0}|L_{D,0},\phi_{0}(z_Y)}|L_{D,0},\phi_{0}(z_Y)])\\
        &=\mathbb{E}[\rho^{L_D}_0\frac{I(\phi_0(z_Y))}{\mathbb{P}(\phi_0(z_Y)|L_{D,0})}(h^D_{1}(h^{Y}_{1}(T^{(s)}_{2}))-h^{L_Y}_{0}(h^D_{1}(h^{Y}_{1}(T^{(s)}_{2}))))g_{L_{Y,0}|L_{D,0},\phi_{0}(z_Y)}]\\
        &=\mathbb{E}[\rho^{L_D}_0\frac{I(\phi_0(z_Y))}{\mathbb{P}(\phi_0(z_Y)|L_{D,0})}(h^D_{1}(h^{Y}_{1}(T^{(s)}_{2}))-h^{L_Y}_{0}(h^D_{1}(h^{Y}_{1}(T^{(s)}_{2}))))g],\\
        \eqref{eq:PortDerivAux4}&=\mathbb{E}[(h^{L_Y}_{0}(h^D_{1}(h^{Y}_{1}(T^{(s)}_{2})))-T_1^{(s)})g_{L_{D,0}|\phi_{0}(z_D)}|\phi_{0}(z_D)]\\
        &=\mathbb{E}[\frac{I(\phi_{0}(z_D))}{\mathbb{P}(\phi_{0}(z_D))}(h^{L_Y}_{0}(h^D_{1}(h^{Y}_{1}(T^{(s)}_{2})))-T_1^{(s)})g_{L_{D,0}|\phi_{0}(z_D)}]\\
        &=\mathbb{E}[\frac{I(\phi_{0}(z_D))}{\mathbb{P}(\phi_{0}(z_D))}(h^{L_Y}_{0}(h^D_{1}(h^{Y}_{1}(T^{(s)}_{2})))-T_1^{(s)})g],
    \end{align*}
    where we have defined 
    \begin{align*}
        \rho^D_s&=\frac{\mathbb{P}(\phi_s(z_D)|D_s=Y_{s-1}=0,\overline{L}_{s-1})}{\mathbb{P}(\phi_s(z_D)|D_{s-1}=Y_{s-1}=0,\overline{L}_{s-1})},\\ \rho^Y_s&=\frac{\mathbb{P}(\phi_s(z_Y)|D_{s}=Y_{s}=0,\overline{L}_{s-1})}{\mathbb{P}(\phi_s(z_Y)|D_{s}=Y_{s-1}=0,\overline{L}_{s-1})},\\
        \rho_s^{L_Y}&=\frac{\mathbb{P}(\phi_s(z_Y)|D_{s}=Y_{s}=0,\overline{L}_{s})}{\mathbb{P}(\phi_s(z_Y)|D_{s}=Y_{s}=0,\overline{L}_{s-1},L_{D,s})},\\\rho_s^{L_D}&=\frac{\mathbb{P}(\phi_s(z_D)|D_{s}=Y_{s}=0,\overline{L}_{s-1},L_{D,s})}{\mathbb{P}(\phi_s(z_D)|D_{s}=Y_{s}=0,\overline{L}_{s-1})}.
    \end{align*}
    Using these notations and the recursive structure given in equation \eqref{eq:PortDerivRecursive}, attending to Definition \ref{def:IF} we can express $\nu_s$'s influence function $\nu_s^1$ as:
    \begin{align*}
        &W_s(1-D_s)(1-Y_s)\left\{\rho^{L_D}_s\rho^{L_Y}_s\rho^D_{s+1}\frac{I(\phi_{s+1}(z_Y))(1-D_{s+1})}{\mathbb{P}(\phi_{s+1}(z_Y)|D_{s+1}=Y_s=0,\overline{L}_s)}(Y_{s+1}-h^{*Y}_{s+1}(1))\right.\nonumber\\
            &+\rho^{L_D}_s\frac{I(\phi_{s}(z_Y))}{\mathbb{P}(\phi_{s}(z_Y)|D_{s}=Y_s=0,\overline{L}_{s-1},L_{D,s})}(h^D_{s+1}(h^{*Y}_{s+1}(1))-h^{L_Y}_{s}(h^D_{s+1}(h^{*Y}_{s+1}(1)))) \nonumber\\
            &+\rho^{L_D}_s\rho^{L_Y}_s\frac{I(\phi_{s+1}(z_D))}{\mathbb{P}(\phi_{s+1}(z_D)|D_{s}=Y_s=0,\overline{L}_s)}(h^{*Y}_{s+1}(1)(1-D_{s+1})-h^D_{s+1}(h^{*Y}_{s+1}(1)))\nonumber\\
            &\left.+\frac{I(\phi_{s}(z_D))}{\mathbb{P}(\phi_{s}(z_D)|D_{s}=Y_s=0,\overline{L}_{s-1})}(h^{L_Y}_{s}(h^D_{s+1}(h^{*Y}_{s+1}(1)))-T^{(s+1)}_{s+1})\right\}\nonumber\\
            &+\sum_{j=1}^s\left[ W_{j-1}(1-D_{j-1})(1-Y_{j-1})\left\{\rho^{L_D}_{j-1}\rho^{L_Y}_{j-1}\rho^D_{j}\frac{I(\phi_{j}(z_Y))(1-D_{j})}{\mathbb{P}(\phi_{j}(z_Y)|D_{j}=Y_{j-1}=0,\overline{L}_{j-1})}\right.\right.\\&\times(T^{(s+1)}_{j+1}(1-Y_{j})-h^{Y}_{j}(T^{(s+1)}_{j+1}))\nonumber\\
            &+\rho^{L_D}_{j-1}\frac{I(\phi_{j-1}(z_Y))}{\mathbb{P}(\phi_{j-1}(z_Y)|D_{j-1}=Y_{j-1}=0,\overline{L}_{j-2},L_{D,j-1})}\\&\times(h^D_{j}(h^{Y}_{j}(T^{(s+1)}_{j+1}))-h^{L_Y}_{j-1}(h^D_{j}(h^{Y}_{j}(T^{(s+1)}_{j+1}))))\nonumber\\
            &+\rho^{L_D}_{j-1}\rho^{L_Y}_{j-1}\frac{I(\phi_{j}(z_D))}{\mathbb{P}(\phi_{j}(z_D)|D_{j-1}=Y_{j-1}=0,\overline{L}_{j-1})}\\&\times(h^{Y}_{j}(T^{(s+1)}_{j+1})(1-D_{j})-h^D_{j}(h^{Y}_{j}(T^{(s+1)}_{j+1})))\nonumber\\
            &+\left.\left.\frac{I(\phi_{j-1}(z_D))}{\mathbb{P}(\phi_{j-1}(z_D)|D_{j-1}=Y_{j-1}=0,\overline{L}_{j-2})}(h^{L_Y}_{j-1}(h^D_{j}(h^{Y}_{j}(T^{(s+1)}_{j+1})))-T^{(s+1)}_{j})\right\}\right],
    \end{align*} 
    having defined the pre-factor $$W_j=\prod_{t=1}^j \rho^{L_D}_{t-1}\rho^{L_Y}_{t-1}\rho^D_{t}\rho^Y_t.$$ The only step remaining in the proof is that the factors which accompany each of the deviations in $\nu^1_s$ correspond to the representation in terms of the risk ratios and the propensity scores given in equation \eqref{eq:InfluenceFunction}. For the first one we have:
    \begin{align*}
        &\rho^Y_{s-1}\rho^{L_D}_{s-1}\rho^{L_Y}_{s-1}\underbrace{\frac{\mathbb{P}(\phi_s(z_D)|D_s=Y_{s-1}=0,\overline{L}_{s-1})}{\mathbb{P}(\phi_s(z_D)|D_{s-1}=Y_{s-1}=0,\overline{L}_{s-1})}}_{\rho^D_{s}}\frac{1}{\mathbb{P}(\phi_{s}(z_Y)|D_{s}=Y_{s-1}=0,\overline{L}_{s-1})}\\
        =&\rho^Y_{s-1}\rho^{L_D}_{s-1}\underbrace{\frac{\mathbb{P}(\phi_{s-1}(z_Y)|D_{s-1}=Y_{s-1}=0,\overline{L}_{s-1})}{\mathbb{P}(\phi_{s-1}(z_Y)|D_{s-1}=Y_{s-1}=0,\overline{L}_{s-2},L_{D,s-1})}}_{\rho^{L_Y}_{s-1}}\times\\&\frac{RR(D_s)}{\mathbb{P}(\phi_{s}(z_Y)|D_{s-1}=Y_{s-1}=0,\overline{L}_{s-1})}\\
        =&\rho^Y_{s-1}\underbrace{\frac{\mathbb{P}(\phi_{s-1}(z_D)|D_{s-1}=Y_{s-1}=0,\overline{L}_{s-2},L_{D,s-1})}{\mathbb{P}(\phi_{s-1}(z_D)|D_{s-1}=Y_{s-1}=0,\overline{L}_{s-2})}}_{\rho^{L_D}_{s-1}}\times\\&\frac{RR(D_s)}{\mathbb{P}(\phi_{s-1}(z_Y)|D_{s-1}=Y_{s-1}=0,\overline{L}_{s-2},L_{D,s-1})\pi_s(z_Y)}\\
        =&\underbrace{\frac{\mathbb{P}(\phi_{s-1}(z_Y)|D_{s-1}=Y_{s-1}=0,\overline{L}_{s-2})}{\mathbb{P}(\phi_{s-1}(z_Y)|D_{s-1}=Y_{s-2}=0,\overline{L}_{s-2})}}_{\rho^Y_{s-1}}\frac{RR(L_{D,s-1})RR(D_s)}{{\mathbb{P}(\phi_{s-1}(z_D)|D_{s-1}=Y_{s-1}=0,\overline{L}_{s-2})}\pi_s(z_Y)}\\
        =&\frac{1}{\mathbb{P}(\phi_{s-1}(z_Y)|D_{s-1}=Y_{s-2}=0,\overline{L}_{s-2})}\frac{RR(L_{D,s-1})RR(D_s)}{\pi_s(z_Y)},
    \end{align*}
    and therefore, applying this argument recursively, we see that $$W_{j-1}\rho^{L_D}_{j-1}\rho^{L_Y}_{j-1}\rho^D_{j}\frac{1}{\mathbb{P}(\phi_{j}(z_Y)|D_{j}=Y_{j-1}=0,\overline{L}_{j-1})}=\Omega_{j-1}\frac{RR(D_j)}{\pi_j(z_Y)}.$$
    Analogously one can show the other equalities needed: 
    \begin{align*}
        W_{j-1}\rho^{L_D}_{j-1}\rho^{L_Y}_{j-1}\frac{1}{\mathbb{P}(\phi_{j}(z_D)|D_{j-1}=Y_{j-1}=0,\overline{L}_{j-1})}&=\Lambda_{j-1}\frac{RR(L_{Y,j-1})}{\pi_j(z_D)},\\
        W_{j-1}\rho^{L_D}_{j-1}\frac{1}{\mathbb{P}(\phi_{j-1}(z_Y)|D_{j-1}=Y_{j-1}=0,\overline{L}_{j-2},L_{D,j-1})}&=\Omega_{j-1},\\
        W_{j-1}\frac{1}{\mathbb{P}(\phi_{j-1}(z_D)|D_{j-1}=Y_{j-1}=0,\overline{L}_{j-2})}&=\Lambda_{j-1}.
    \end{align*}
    This effectively concludes the proof, as at this point we can rewrite $\nu_s^1$ as each of the $s$-addends in equation \eqref{eq:IdentificationFormula}.
\end{proof}

Recall that the influence function reflects the sensitivity of the estimand to misspecifications in the distribution of the process. We see how each of the terms in Equation \eqref{eq:InfluenceFunction} represents the deviation incurred at each step when following the variables of the system which are not intervened on in the temporal order in which they are observed. In addition, the bias correction granted by the influence function to the simple plug-in estimator to achieve doubly robustness is related to the weighted estimators described in Section \ref{sec:WeightedEstimators} and Appendix \ref{sec:AppEstimationWeighted}. Note how in each of the deviations in the influence function \eqref{eq:InfluenceFunction} associated with $Y$ or $L_Y$ we set $Z=z_Y$, correct for the distributions of $D, L_D$, and perform inverse propensity score weighting for the intervened variables, as it was done in $\widehat{\nu}_{weighted,Y}$. These qualities of the influence function yield doubly robustness of the one-step estimator, as stated in the next theorem.
\begin{restatable}{theorem}{thmDR}\label{thm:DR}
    The one-step estimator $\widehat{\mathbb{E}}_n[\nu^1(\Tilde{\mathbb{P}}_n)]+\nu(\Tilde{\mathbb{P}}_n)$ is a doubly robust estimator of the  probability of observing the event of interest under a sustained treatment strategy, meaning that under the identification conditions of Theorem \ref{thm:IdentifcationFormula} as long as the model for the propensity score $\pi_j(z)$ is correctly specified and consistently estimated, then said estimator is consistent for the  probability of  interest if either the models for the conditional probabilities of $Y, L_Y$ or  $D, L_D$ are correctly specified and consistently estimated, but not necessarily both.
\end{restatable}

\begin{proof}
    As the postulated models are consistently estimated, we can work directly with $\Tilde{\mathbb{P}}$. Then, the one-step estimator can be seen as the solution to the estimating equations $\widehat{\mathbb{E}}_n[\nu^1(\Tilde{\mathbb{P}})+\nu(\Tilde{\mathbb{P}})-\widehat{\nu}_{DR}]=0$. Therefore, to show consistency, as $\nu(\Tilde{\mathbb{P}})$ is a deterministic quantity, it suffices to show that $\mathbb{E}[\nu^1(\Tilde{\mathbb{P}})]=\nu-\nu(\Tilde{\mathbb{P}})$, or equivalently that $\mathbb{E}[\nu^1_s(\Tilde{\mathbb{P}})]=\nu_s-\nu_s(\Tilde{\mathbb{P}})$ for all $s=0,\ldots,K$. We will show that this holds when the models for $Y,L_Y, C , R$  are correctly specified, but not necessarily those for $D,L_D$. The reverse case follows analogously. We study each of the four types of addends we have in equation \eqref{eq:InfluenceFunction}. Firstly:
    \begin{align*}
        &\mathbb{E}[\Tilde{\Omega}_{j-1}(1-Y_{j-1})\frac{\Tilde{RR}(D_{j})}{\pi_{j}(z_Y)}(1-D_{j})I(\phi_{j}(z_Y))(\Tilde{T}^{(s+1)}_{j+1}(1-Y_{j})-h^{Y}_{j}(\Tilde{T}^{(s+1)}_{j+1}))]\\
        &=\mathbb{E}[\Tilde{\Omega}_{j-1}(1-Y_{j-1})\frac{\Tilde{RR}(D_{j})}{\pi_{j}(z_Y)}(1-D_{j})I(\phi_{j}(z_Y))\\&\quad\times\mathbb{E}[\Tilde{T}^{(s+1)}_{j+1}(1-Y_{j})-h^{Y}_{j}(\Tilde{T}^{(s+1)}_{j+1})|\overline{D}_j,\overline{Y}_{j-1},\overline{L}_{j-1},Z,\overline{C}_j,\overline{R}_{j-1}]]\\
        &=\mathbb{E}[\Tilde{\Omega}_{j-1}(1-Y_{j-1})\frac{\Tilde{RR}(D_{j})}{\pi_{j}(z_Y)}(1-D_{j})I(\phi_{j}(z_Y))\\&\quad\times\underbrace{\mathbb{E}[\Tilde{T}^{(s+1)}_{j+1}(1-Y_{j})-h^{Y}_{j}(\Tilde{T}^{(s+1)}_{j+1})|{D}_j={Y}_{j-1}=0,\overline{L}_{j-1},\phi_j(z_Y)]}_{=0\text{ a.s. w.r.t. $\overline{L}_{j-1}$ as model for $Y$correct}}]=0,
    \end{align*}
    where the $\Tilde{h}$ functions represent the $h$ functions with expectations taken under $\Tilde{\mathbb{P}}$, and \'idem for $T_j^{(s)}$. In the same way:
    \begin{align*}
        &\mathbb{E}[\Tilde{\Omega}_{j-1}(1-D_{j-1})(1-Y_{j-1})I(\phi_{j-1}(z_Y))(\Tilde{h}^D_{j}(h^{Y}_{j}(\Tilde{T}^{(s+1)}_{j+1}))-h^{L_Y}_{j-1}(\Tilde{h}^D_{j}(h^{Y}_{j}(\Tilde{T}^{(s+1)}_{j+1}))))]\\
        &=\mathbb{E}[\Tilde{\Omega}_{j-1}(1-D_{j-1})(1-Y_{j-1})I(\phi_{j-1}(z_Y))\mathbb{E}[\Tilde{h}^D_{j}(h^{Y}_{j}(\Tilde{T}^{(s+1)}_{j+1}))\\&\qquad-h^{L_Y}_{j-1}(\Tilde{h}^D_{j}(h^{Y}_{j}(\Tilde{T}^{(s+1)}_{j+1})))|\overline{D}_{j-1},\overline{Y}_{j-1},L_{D,j-1},\overline{L}_{j-2},Z,\overline{C}_{j-1},\overline{R}_{j-2}  ]]\\
        &=\mathbb{E}[\Tilde{\Omega}_{j-1}(1-D_{j-1})(1-Y_{j-1})I(\phi_{j-1}(z_Y))\mathbb{E}[\Tilde{h}^D_{j}(h^{Y}_{j}(\Tilde{T}^{(s+1)}_{j+1}))\\&\qquad-h^{L_Y}_{j-1}(\Tilde{h}^D_{j}(h^{Y}_{j}(\Tilde{T}^{(s+1)}_{j+1})))|{D}_{j-1}={Y}_{j-1}=0,L_{D,j-1},\overline{L}_{j-2},\phi_{j-1}(z_Y)  ]]\\
        &=0.
    \end{align*}
    The remaining two terms will not have expectation zero. Indeed:
    \begin{align*}
        &\mathbb{E}[\Lambda_{j-1}(1-D_{j-1})(1-Y_{j-1})\frac{RR(L_{Y,j-1})}{\pi_{j}(z_D)}I(\phi_{j}(z_D))\bigg(h^{Y}_{j}(\Tilde{T}^{(s+1)}_{j+1})(1-D_{j})\\&\qquad-\Tilde{h}^D_{j}(h^{Y}_{j}(\Tilde{T}^{(s+1)}_{j+1}))\bigg)]\\
        &=\mathbb{E}[\Lambda_{j-1}(1-D_{j-1})(1-Y_{j-1})RR(L_{Y,j-1})I(\phi_{j-1}(z_D))\\&\quad\times\mathbb{E}[\frac{I(C_j=0,R_{j-1}=1)}{\pi_j(z_D)}(h^{Y}_{j}(\Tilde{T}^{(s+1)}_{j+1})(1-D_{j})\\&\quad-\Tilde{h}^D_{j}(h^{Y}_{j}(\Tilde{T}^{(s+1)}_{j+1})))|{D}_{j-1}={Y}_{j-1}=0,\overline{L}_{j-1},\phi_{j-1}(z_Y),C_j,R_{j-1}   ]]\\
        &=\mathbb{E}[\Lambda_{j-1}(1-D_{j-1})(1-Y_{j-1})RR(L_{Y,j-1})I(\phi_{j-1}(z_D))\\&\quad\times\mathbb{E}[(h^{Y}_{j}(\Tilde{T}^{(s+1)}_{j+1})(1-D_{j})-\Tilde{h}^D_{j}(h^{Y}_{j}(\Tilde{T}^{(s+1)}_{j+1})))|{D}_{j-1}={Y}_{j-1}=0,\overline{L}_{j-1},\phi_{j}(z_Y)   ]]\\
        &=h_0^{L_D}(\ldots h_{j-1}^{L_D}(h_{j-1}^{L_Y}(h_{j}^{D}(h_{j}^{Y}(\Tilde{T}^{(s+1)}_{j+1}))))\ldots)\\&\quad-h_0^{L_D}(\ldots h_{j-1}^{L_D}(h_{j-1}^{L_Y}(\Tilde{h}_{j}^{D}(h_{j}^{Y}(\Tilde{T}^{(s+1)}_{j+1}))))\ldots),
    \end{align*}
    and proceeding the same way for the last term:
    \begin{align*}
        &\mathbb{E}[\Lambda_{j-1}(1-D_{j-1})(1-Y_{j-1})I(\phi_{j-1}(z_D))(h^{L_Y}_{j-1}(\Tilde{h}^D_{j}(h^{Y}_{j}(\Tilde{T}^{(s+1)}_{j+1})))-\Tilde{T}^{(s+1)}_{j})]\\
        &=h_0^{L_D}(\ldots h_{j-1}^{L_D}(h_{j-1}^{L_Y}(\Tilde{h}_{j}^{D}(h_{j}^{Y}(\Tilde{T}^{(s+1)}_{j+1}))))\ldots)-h_0^{L_D}(\ldots h_{j-1}^{Y}(\Tilde{T}^{(s+1)}_{j})\ldots).
    \end{align*}
    When we  add these last two equations we see that the crossed terms cancel out, and after we sum over $j=1,\ldots,s$ and also add the term in $s+1$, we get a telescopic sum, which equals the last term of the first addend minus the first term of the second addend. As a result: 
    $$\mathbb{E}[\nu^1_s(\Tilde{\mathbb{P}})]=h_0^{L_D}(\ldots h_{s}^{L_D}(h_{s}^{L_Y}(h_{s+1}^{D}(h_{s+1}^{*Y}(1))))\ldots)-\Tilde{T}^{(s+1)}_{1}=\nu_s-\nu_s(\Tilde{\mathbb{P}}).$$ This concludes the proof.
\end{proof}

\section{Implications to related work.}\label{sec:AppRelatedWork}
\subsection{A connection to general interventionist mediation analysis}\label{sec:ConnectionInterventionistMediation}
General interventionist mediation analysis  is closely related to the field of separable effects and sustained treatment strategies. For example, analogues to the dismissible component conditions were described in the context of mediation analysis \cite{didelez2019defining}, inspired by \cite{robins2010alternative}. Both of these works considered point treatments.  \cite{robins2020interventionist} mentions, as a potential development of their work, the extension to multiple treatments $Z_1,\ldots,Z_k$, each of which could be decomposed into two components following a generalized decomposition assumption as the one we introduced in Section \ref{sec:TreatmentDecomposition} (see \cite[Sec. 3.5 (a)]{robins2020interventionist} for further details). The identification of the distribution of the process under interventions on the components of said treatment variables was left as an open problem. 

We conjecture  that the formalism we present in this work can be applied to this open problem. Consider the multiple treatment variables $Z_1,\ldots,Z_k$ represent treatment at time $k$, and that (following the notation of \cite{robins2020interventionist}) each $Z_i$ is decomposed as $(O_i,N_i)$ such that in the observed data $O_i\equiv Z_i\equiv N_i$. The goal is identification under an intervention which sets $(O_1,N_1,\ldots,O_k,N_k)=(o_1,n_1,\ldots,o_k,n_k)$. To accommodate this  treatment strategy into our formalism, we would not consider a time-varying treatment, but rather define the strategy-centered adherence indicator in the four arm trial at each time point based on the strategy $(o_1,n_1,\ldots,o_k,n_k)$. A perfect adherence intervention is then equivalent to the given treatment strategy, and would allow for identification using data coming only from the two-arm trial. 

This speaks to one of the key points of our framework and the strategy-centered encoding: understanding a certain mediation strategy as perfectly  following a protocol means we can encode perfect adherence without decomposing treatment at every time point, but only at baseline, while still capturing all the perfect adherence information by a single encoder at each time point.

\subsection{Extending the doubly robust estimator to estimands in related work}\label{sec:AppRelatedWorkDR}
Other authors \cite{stensrud2021generalized,stensrud2022separable,wanis2024separable} have considered time-varying settings, sometimes with competing events, which greatly resemble or are special cases of the one considered here. While they derived identification results and g-formulas for their counterfactual quantities of interest and propose some estimators for them, they have not derived one-step (doubly robust) estimators for these quantities. We reason now that the one-step estimator presented in this section can be easily adapted to their scenarios, yielding a doubly robust estimator of the quantities they study. This was not done in any of the works we now mention:

\begin{itemize}
    \item In \cite{stensrud2021generalized} they consider a time-varying setting with competing events and time-varying covariates. Exactly our setting, but without accounting for adherence. In fact, said work constitutes the intention-to-treat counterpart to the one we present here. Their g-formula (Theorem 1 in Appendix B) is exactly the one we derive in \eqref{eq:IdentificationFormula} if we dropped the adherence indicators $R$. Thus, if we dropped $\overline{R}_{K+1}$, Theorem \ref{thm:InfluenceFucntion} gives us the influence function for the estimand in \cite[Thm. 1, App. B]{stensrud2021generalized}, and most importantly Theorem \ref{thm:DR} will still hold, meaning the one-step estimator based in this influence function would be doubly robust, in the sense of this theorem.
    \item \cite{stensrud2022separable} considered our setting without time-varying covariates. Our identification formula reduces to their Equation (10) when $\overline{L}_K=L_0$ and $L_0\independent Z$. Therefore, under this simplification Theorem \ref{thm:InfluenceFucntion} gives us the influence function for the estimand in \cite[Eq. (10)]{stensrud2021generalized}, and Theorem \ref{thm:DR} still applies. Note that in this scenario the influence function greatly simplifies, as the $T$ functions do not involve conditional expectations on $L_k$, and all the risk ratios for the time-varying covariates are trivially one. 
    \item The framework considered in \cite{wanis2024separable} can as well be reduced to ours. We could think of their time-varying adherence indicator as our competing event $D$. The main difference is that their adherence indicator does not need to be zero in order for the process to continue. This is reflected in the fact that in \cite[Eq. (3)]{wanis2024separable} they sum over all possible adherence histories. Furthermore, they do not consider censoring. Hence, if we dropped our censoring and adherence indicators $C,R$ and summed over all possible histories of $\overline{D}_{K+1}$ in \eqref{eq:IdentificationFormula} we would retrieve the identification formula in \cite{wanis2024separable}. The influence function for their estimand would then be obtained from the one in Theorem \ref{thm:InfluenceFucntion} by dropping $C,R$, deleting all the $(1-D)$ factors from $\nu^1$ and the $h^D$ functions, and having all risk ratios and $h$ functions be evaluated at the $D$-history of each subject. Note how as we drop $C,R$, the propensity scores $\pi_j(z)$ become trivially one.
\end{itemize}
This shows the generality and versatility of the work we present here, how it can be seen to generalize a wide range of existing work, and provides a comprehensive framework for the analysis of time-varying setting with competing events.

\section{The simulation study presented in Section \ref{sec:SimulationStudy}.}\label{sec:AppSimus}

The data generating mechanism in the four-arm trial under the strategy-centered encoding which we study in Section  \ref{sec:SimulationStudy} is:
\begin{align*}\allowdisplaybreaks
     Z_D,&\  Z_Y\sim Ber(1/2)\text{ independent, }L_0\sim Ber(1/2),\quad C_1\sim Ber(1/50),\\
     R_1&\mid C_1=0, L_0 \sim Ber(4/5),\\
    D_1&\mid C_1=0,R_1,L_0,Z_D,Z_Y\sim \begin{cases}
        Ber((1+L_0)/20),\quad Z_D=R_1\\Ber((1+L_0)/30),\quad Z_D\neq R_1
    \end{cases},\\
    Y_1&\mid C_1=D_1=0,R_1,L_0,Z_Y,Z_D\sim \begin{cases}
        Ber((10+2L_0)/20),\quad Z_Y=R_1\\Ber((10-2L_0)/20),\quad Z_Y\neq R_1
    \end{cases},\\
    L_1&\mid C_1=D_1=Y_1=0,R_1,L_0,Z_Y,Z_D\sim \begin{cases}
        Ber((2+L_0)/4),\quad Z_D=R_1\\Ber((1+L_0)/4),\quad Z_D\neq R_1
    \end{cases},\\
    C_2&\mid C_1=D_1=Y_1=0,R_1,\overline{L}_1,Z_Y,Z_D\sim Ber((1+L_1)/20),\\
    R_2&\mid C_2=D_1=Y_1=0,R_1,\overline{L}_1,Z_Y,Z_D\sim Ber((3+L_1)/5),\\
    D_2&\mid C_2=D_1=Y_1=0,\overline{R}_2,\overline{L}_1,Z_Y,Z_D\sim \begin{cases}
        Ber((1+L_1)/20),\quad Z_D=R_2\\Ber((1+L_1)/30),\quad Z_D\neq R_2
    \end{cases},\\
    Y_2&\mid C_2=D_2=Y_1=0,\overline{R}_2,\overline{L}_1,Z_Y,Z_D\sim \begin{cases}
        Ber((10+2L_1)/20),\quad Z_Y=R_2\\Ber((10-2L_1)/20),\quad Z_Y\neq R_2
    \end{cases},
\end{align*}
where $Ber(p)$ denotes a Bernoulli distribution with probability of success $p\in(0,1)$. Note that $Z=R_k$ indicates that the  individual takes treatment 1 at time $k$, meaning we define effects in terms of treatment actually taken. The  causal DAG of this process can be seen in Figure \ref{fig:DagSimulationStudy}.
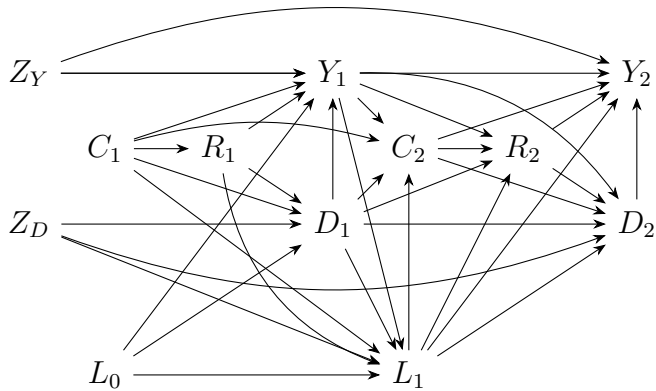
\begin{figure}
    \centering
\begin{tikzpicture}
\begin{scope}[every node/.style={thick,draw=none}]
    \node (ZD) at (-2,-1) {$Z_D$};
    \node (ZY) at (-2,1) {$Z_Y$};
    \node (L0) at (-1,-3) {$L_0$};
    \node (C1) at (-1,0) {$C_1$};
    \node (R1) at (0.5,0) {$R_1$};
    \node (Y1) at (2,1) {$Y_1$};
    \node (D1) at (2,-1) {$D_1$};
    \node (L1) at (3,-3) {$L_1$};
    \node (C2) at (3,0) {$C_2$};
    \node (R2) at (4.5,0) {$R_2$};
    \node (Y2) at (6,1) {$Y_2$};
    \node (D2) at (6,-1) {$D_2$};
\end{scope}

\begin{scope}[>={Stealth[black]},
              every node/.style={fill=white,circle},
              every edge/.style={draw=black}]
    \path [->] (ZY) edge  (Y1);
    \path [->] (ZD) edge  (D1);
    \path [->] (ZY) edge  (Y1);
    \path [->] (R1) edge  (D1);
    \path [->] (R1) edge  (Y1);
    \path [->] (L0) edge  (D1);
    \path [->] (L0) edge  (Y1);
    \path [->] (L1) edge  (C2);
    \path [->] (L1) edge  (R2);
    \path [->] (L1) edge  (D2);
    \path [->] (L1) edge  (Y2);
    \path [->] (L0) edge  (L1);
    \path [->] (R1) edge[bend right]  (L1);
    \path [->] (ZD) edge  (L1);
    \path [->] (ZY) edge[bend left=20]  (Y2);
    \path [->] (ZD) edge[bend right=20]  (D2);
    \path [->] (R2) edge  (D2);
    \path [->] (R2) edge  (Y2);
    \path [->] (C1) edge  (R1);
    \path [->] (C1) edge  (Y1);
    \path [->] (C1) edge  (D1);
    \path [->] (C1) edge  (L1);
    \path [->] (D1) edge  (Y1);
    \path [->] (D1) edge  (L1);
    \path [->] (D1) edge  (D2);
    \path [->] (D1) edge  (C2);
    \path [->] (D1) edge  (R2);
    \path [->] (C1) edge[bend left=15]  (C2);
    \path [->] (Y1) edge  (C2);
    \path [->] (Y1) edge  (Y2);
    \path [->] (Y1) edge[bend left=27.5]  (D2);
    \path [->] (Y1) edge  (R2);
    \path [->] (Y1) edge  (L1);
    \path [->] (C2) edge  (R2);
    \path [->] (C2) edge  (Y2);
    \path [->] (C2) edge  (D2);
    \path [->] (D2) edge  (Y2);
\end{scope}
\end{tikzpicture}

\caption{\small Causal DAG \cite{robins2010alternative} of the data generating process considered in the simulation study conducted in Section \ref{sec:SimulationStudy}.}
\label{fig:DagSimulationStudy}
\end{figure}

\section{Analysis of the SPRINT}\label{sec:AppSprint}
\subsection{Adherence in the SPRINT cohort}\label{sec:AppSprintAdherence}

As explained in Section \ref{sec:SprintExample} we consider  $R_k=1$ when the individual said they 100\% followed their assigned treatment at the clinical visit at time $k+1$. To have a deeper understanding of the adherence behaviour of the SPRINT cohort, in Figure \ref{fig:SprintAdherence} we plotted the fraction of individuals at risk at time $k$ (meaning individuals that reached time-step $k$ and $C_k=0$) which have a perfect adherence pattern ($\overline{R}_k=1$). We see how adherence is initially greater that 90\% in both arms, and that  decreases with follow-up time. Notably, adherence is lower in the intensive blood pressure therapy arm, which might be due to the harsher side effects as a consequence of the most aggressive treatment regime.

It is worth pointing out that the fraction of perfect adherers at the first time point is not 1. This means there are some patients who will be artificially censored at time $k=1$, as $R_1=0$. This is not a problem under  conditional exchangeability (Assumption \ref{ass:Exchangeability}). Furthermore, if we perform a test of equality of proportions \cite{newcombe1998interval}, we see that at a 5\% significance level we cannot reject the null hypothesis that $\mathbb{P}(R_1=1\mid Z=1)=\mathbb{P}(R_1=1\mid Z=0)$, meaning both trial populations have identical baseline adherence. However, we clearly see in Figure \ref{fig:SprintAdherence} that this is not the case as the trial goes on.

\begin{figure}
    \centering
    \includegraphics[width=0.75\linewidth]{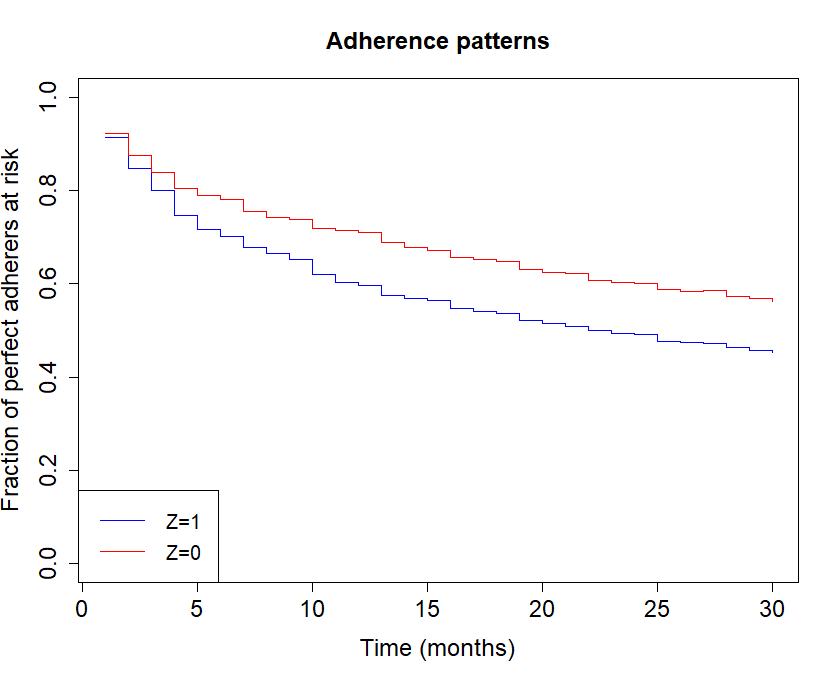}
    \caption{Adherence pattern in each treatment arm of the SPRINT, represented as the fraction of individuals at risk at each time point which have a perfect adherence history, over the follow-up period of 30 months.}
    \label{fig:SprintAdherence}
\end{figure}

\subsection{The fitted models}\label{sec:AppSprintModels}
As explained in Section \ref{sec:SprintExample}, under the assumption that the identification assumptions of Theorem \ref{thm:IdentifcationFormula} hold when $L_{Y,k}=\emptyset$, this weighted estimator introduced in Section \ref{sec:WeightedEstimators} requires modelling of the  conditional distributions of $\Bar{Y}_{K+1}, \Bar{C}_{K+1},\Bar{R}_{K+1}$. Note how this estimator is less computationally demanding compared to the simple and one-step estimator, as these require the computations of $(K+1)$-dimensional integrals with respect to the model of $\overline{L}_K$

To compute the weighted estimator, we fitted pooled logistic regression models for each treatment population. We denote the log-mean arterial pressure at baseline by $L_0^{MAP}$, and all other (categorical) baseline covariates by $L_0^*$. For $Z=1$ and $k$ in $\{0,\ldots,K\}$ we choose
\begin{align*}
    logit(\mathbb{P}(Y_{k+1}=1 \mid D_{k+1}=Y_{k}= 0, \overline{L}_{k}, \phi_{k+1}(1)))&=\beta_{1,L_0^*}+\beta_2 k + \beta_3 k^2 + \beta_4 k^3 \\&+\beta_5 L_k,\\
    logit(\mathbb{P}(C_{k+1}=1\mid D_{k}=Y_{k}=0,\overline{L}_{k},\phi_{k}(1)))&= \delta_{1,L_0^*}+\delta_2 k + \delta_3 k^2 + \delta_4 k^3 \\&+ \delta_5 L_k,\\ 
    logit(\mathbb{P}(R_{k+1}=1\mid D_{k}=Y_{k}=C_{k+1}=0,\overline{L}_{k},\phi_{k}(1)))&= \gamma_{1,L_0^*}+\gamma_2 k + \gamma_3 k^2 \\&+ \gamma_4 k^3 + \gamma_5 L_k,
\end{align*}

where in the case $k=0$ the last addend represents a linear term in $L_0^{MAP}$. For the standard treatment group $Z=0$ we fit models analogous to these. The estimated  probabilities can be seen in Section \ref{sec:SprintResults}.

\end{document}